\documentclass[runningheads,orivec]{llncs}
\usepackage[utf8]{inputenc}

\usepackage{amsmath,amssymb,amsfonts}
\usepackage{algorithmic}
\usepackage{graphicx}
\usepackage{textcomp}

\usepackage{float} 
\usepackage[all]{xy}
\usepackage{epsfig, amssymb, amsmath}
\usepackage{graphics, graphicx}
\usepackage{latexsym,amsfonts,amsmath,amssymb}
\usepackage{hyperref}
\usepackage{url}
\usepackage{relsize}
\usepackage{alltt}
\usepackage{comment}
\usepackage{array}
\usepackage{tikz}
\usepackage{pifont}
\usepackage{multirow}
\usepackage{wasysym}
\usepackage{amssymb}
\usepackage{booktabs}
\usepackage{xcolor,pifont}
\newcommand*\colourcheck[1]{%
  \expandafter\newcommand\csname #1check\endcsname{\textcolor{#1}{\ding{52}}}%
}
\colourcheck{blue}
\colourcheck{green}
\colourcheck{red}

\long\def\comment#1{}

\newcommand{\caH}{{\mathcal H}}

\newcommand{\caP}{{\mathcal P}}

\newcommand{\bvar}[1]{\mathit{shVar}(#1)}

\newcommand{\wf}{\mathit{wf}}

\newcommand{\chV}[1]{{\textit #1}_?}

\newcommand{\var}[1]{\mathit{Vars}(#1)}

\newcommand{\Symbols}{\Sigma}

\newcommand{\sort}[1]{\ensuremath{\mathsf{#1}}}

\newcommand{\nI}[1]{\ensuremath{#1{\ \notin \ }{\mathcal I}}}
\newcommand{\inI}[1]{\ensuremath{#1{\ \in \ }{\mathcal I}}}

\newcommand{\ProcConf}{\mathit{ProcConf}} 
\newcommand{\LProc}{\mathit{LProc}} 
\newcommand{\ProcId}{\mathit{ProcId}} 
\newcommand{\Role}{\mathit{Role}} 
\newcommand{\Proc}{\mathit{Proc}}

\newcommand{\Msg}{\mathit{Msg}} 
\newcommand{\Real}{\mathit{Real}} 
 
\newcommand{\Nat}{\mathit{Nat}} 
\newcommand{\Cond}{\mathit{Cond}}

\newcommand{\nil}{\mathit{nilP}}

\newcommand{\PSymbols}{\Symbols_\caP}
\newcommand{\SpecPA}{\caP_\mathit{PA}}
\newcommand{\ProcPA}{P_\mathit{PA}}
\newcommand{\PEq}{E_{\caP}}
\newcommand{\PAEq}{E_\mathit{PA}}
\newcommand{\PAPEq}{E_\mathit{\textit{PA}_\caP}}
\newcommand{\PASymbols}{\Symbols_\mathit{{PA}_\caP}}
\newcommand{\SpecTPA}{\caP_\mathit{TPA}}
\newcommand{\ProcTPA}{P_\mathit{TPA}}

\newcommand{\TPAPEq}{E_\mathit{\textit{TPA}_\caP}}
\newcommand{\TPASymbols}{\Symbols_\mathit{{TPA}_\caP}}
\newcommand{\BNFSymbols}{\Symbols_\mathit{PA}}
\newcommand{\TBNFSymbols}{\Symbols_\mathit{TPA}}
\newcommand{\PAStateSymbols}{\Symbols_\mathit{{\textit{PA}_\caP+State}}}
\newcommand{\TPAStateSymbols}{\Symbols_\mathit{{\textit{TPA}_\caP+State}}}
\newcommand{\PARls}{R_{\mathit{\textit{PA}_\caP}}}
\newcommand{\TPARls}{R_{\mathit{\textit{TPA}_\caP}}}

\newcommand{\tpatopa}{\textit{tpa2pa}}
\newcommand{\tpatopax}{\textit{tpa2pa\!*}}

\newcommand{\eqname}[1]{\tag{#1}}
\newcommand{\MaxProcId}{\textit{id}}

\begin{document}

\title{Protocol Analysis with Time
\thanks{This paper was 
partially supported by the EU (FEDER) and the Spanish
MCIU under grant RTI2018-094403-B-C32,
by the Spanish Generalitat Valenciana under grant PROMETEO/2019/098 and APOSTD/2019/127, by the US Air Force Office of Scientific Research 
under award number FA9550-17-1-0286, and by ONR Code 311.}
}

\author{}
\author{
Dami\'an Aparicio-S\'anchez\inst{1}
\and       
Santiago Escobar\inst{1}
\and       
Catherine Meadows\inst{2}
\and       
Jos{\'e} Meseguer\inst{3}
\and Julia Sapi\~{n}a\inst{1}
}
\institute{
Universitat Polit\`ecnica de Val\`encia, Spain\\
       \email{\{daapsnc,sescobar,jsapina\}@upv.es}
\and       
Naval Research Laboratory, Washington DC, USA\\
       \email{meadows@itd.nrl.navy.mil}
\and       
University of Illinois at Urbana-Champaign, USA\\
       \email{meseguer@illinois.edu}
}
\date{}

\maketitle

\begin{abstract}
We present a framework suited to the analysis of cryptographic protocols that make use of time in their execution.   We provide a process algebra syntax that makes time information available to processes, and a   transition semantics that takes account of fundamental properties of time. Additional properties can be added by the user if desirable. This timed protocol framework can be implemented either as a simulation tool or as a symbolic analysis tool
in which time references are represented by logical variables,  and  in which the properties of time are implemented as constraints on those time logical variables.  These constraints are carried along the symbolic execution of the protocol.  The satisfiability of these constraints can be evaluated as the analysis proceeds, so attacks that violate the laws of physics can be rejected as impossible.    
We demonstrate the feasibility of our approach by using the Maude-NPA protocol analyzer  together with an SMT solver that is used  to evaluate the satisfiability of timing constraints. We  provide a sound and complete protocol transformation from our timed process algebra to the Maude-NPA syntax and semantics, and we prove its soundness and completeness.   We then use the tool to  analyze  Mafia fraud and distance hijacking attacks on a suite of  distance-bounding protocols.

 \end{abstract}

\section{Introduction}

Time is an important aspect of many cryptographic protocols, and there has been increasing interest in the formal analysis of protocols that use time. Model checking of protocols that use time can be done using either an explicit time model, or by using an untimed model and showing it is sound and complete with respect to a timed model. The former is more intuitive for the user, but the latter is often chosen because not all cryptographic protocol analysis tools support reasoning about time.  In this paper we describe a solution  that combines the  advantages of both approaches.  An explicit timed specification language is developed with a timed  syntax and semantics, and is automatically translated to an existing untimed language.  The user however writes protocol specifications and queries in the timed language.  In this paper  we describe how such an approach has been applied to the Maude-NPA tool by taking advantage of its built-in support for constraints.  We believe that this approach can be applied to other tools that support  constraint handling as well.  

There are a number of security protocols that make use of time.  In general, there are two types: those that make use of assumptions about time, most often assuming some sort of loose synchronization, and those that guarantee these assumptions.    The first kind includes protocols such as Kerberos~\cite{NYHR05}, which uses timestamps to defend against replay attacks, the TESLA protocol~\cite{PSCT+05}, which relies on loose synchronization to amortize digital signatures, and  blockchain protocols, which use timestamps to order blocks in the chain. 
The other kind provides guarantees based on physical properties of time: for example, distance bounding, which guarantees that a prover is within a certain distance of a verifier, and secure time synchronization, which guarantees that  the clocks of two different nodes are synchronized within a certain margin of error. In this paper, we concentrate on protocols using distance bounding, both because it has been well-studied, and because the timing constraints are relatively simple. 

A number of approaches have been applied to the analysis of distance bounding protocols.  
In \cite{MPPC+07}, an epistemic logic for distance bounding analysis is presented where timing is captured by means of \emph{timed channels}, which are described axiomatically. Time of sending and receiving  messages can be deduced by using these timed channel axioms.  
In \cite{BCSS11}, Basin et al. define a formal model for reasoning about physical properties of security protocols, including timing and location, which they formalize in Isabelle/HOL and use it to analyze several distance bounding protocols, by applying a technique similar to Paulson's inductive approach \cite{Paulson98}.  
In \cite{DD19}, Debant et al. develop a timing model for AKiSS, a tool for verifying protocol equivalence in the bounded session model, and use it to analyze distance bounding protocols.  
In \cite{NTU18}, Nigam et al. develop a model of timing side channels in terms of constraints  and use it to define a timed  version of observational equivalence for protocols.  They have developed a tool for verifying observational equivalence that relies on SMT solvers.   Other work concentrates on simplifying the problem so it can be more easily analyzed by a model checker,  but proving that the simple problem is sound and complete with respect to the original problem so that the analysis is useful.   In this regard, Nigam et al.~\cite{NTU16}  and  Debant et al.~\cite{DDW18} show that it is safe to  limit the size and complexity of the topologies, and Mauw et al.~\cite{MSTT18} and Chothia et al.~\cite{CRS18} develop timed and untimed models and show that analysis in the untimed model is sound and complete with respect to the timed model.   

 In this paper we illustrate our approach  by developing a timed protocol semantics suitable for the analysis of protocols that use constraints on time and distance, such as distance bounding,  and that  can be implemented as either a simulation tool for generating and checking concrete configurations, or as a symbolic analysis tool  that allows the exploration of all relevant configurations.     
We realize the timed semantics by translating it into the semantics of the Maude-NPA protocol analysis tool, in which timing properties are expressed as constraints.  The constraints generated during the Maude-NPA search are then checked using an SMT solver.
  
There are several  things that help us.  
One is that
we consider a metric space with distance constraints.
Many tools support constraint handling,
e.g., Maude-NPA~\cite{EMMS15} and Tamarin~\cite{MSCB13}.
Another is that time can be naturally added to a process algebra.
Many tools support processes, 
e.g., Maude-NPA~\cite{YEMM+16} and AKISS~\cite{DD19}.

The rest of this paper is organized as follows.
In Section \ref{sec:examples}, we recall the Brands-Chaum protocol, which is used as the running example throughout the paper. 
In Section~\ref{sec:process-time}, we present the timed process algebra with its intended semantics. 
In Section~\ref{sec:simple}, we present a sound and complete protocol transformation from our timed process algebra to an untimed process algebra. 
In Section~\ref{sec:byhand}, we show how our timed process algebra 
can be transformed into Maude-NPA strand notation. 
In Section~\ref{sec:experiments}, we present our experiments. 
We conclude in Section~\ref{sec:conclusions}.

\section{The Brands-Chaum distance bounding protocol}\label{sec:examples}

In the following, we recall the Brands-Chaum distance bounding protocol of \cite{BC93}, which we will use as the running example for the whole paper.

\begin{example}\label{ex:brands-and-chaum}
The Brands-Chaum protocol specifies communication  between a verifier V and a prover P.  P needs to authenticate itself to V, and also needs to prove that it is within a distance ``d" of it. 
$X {;} Y$ denotes concatenation of two messages $X$ and $Y$,
$\textit{commit}(N,Sr)$ denotes commitment of  secret $Sr$ with a nonce $N$,  
$\textit{open}(N,Sr,C)$ denotes opening a commitment $C$ using the nonce $N$ and checking whether it carries the secret $Sr$,
$\oplus$ is the exclusive-or operator, and $\textit{sign}(A,M)$ denotes  $A$ signing message $M$. A typical interaction between the prover and the verifier is as follows:

\noindent
{\small%
\begin{align*}
P \rightarrow V &:\textit{commit}(N_P,S_P) \\[-1ex]
& \mbox{//The prover sends his name and a commitment}\\
V \rightarrow P &:N_V \\[-1ex]
& \mbox{//The verifier sends a nonce}\\[-1ex]
& \mbox{//and records the time when this message was sent}\\[-1ex]
P \rightarrow V &: N_P \oplus N_V \\[-1ex]
& \mbox{//The verifier checks the answer of this exclusive-or}\\[-1ex]
& \mbox{//message arrives within two times a fixed distance}\\
P \rightarrow V &: S_P \\[-1ex]
& \mbox{//The prover sends the committed secret}\\[-1ex]
& \mbox{//and the verifier checks $\textit{open}(N_P,S_P,\textit{commit}(N_P,S_P))$}\\
P \rightarrow V &: \mathit{sign}_P(N_V ; N_P \oplus N_V) \\[-1ex]
& \mbox{//The prover signs the two rapid exchange messages}
\end{align*}
}
\end{example}

\noindent
The previous informal Alice\&Bob notation can be naturally extended to include time.
We consider wireless communication between the participants located at an arbitrary given topology 
(participants do not move from their assigned locations)
with distance constraints, 
where time and distance are equivalent for simplification and are represented by a real  number.
We  assume a metric space with a distance function $d: A \times A \to \Real$ from a set $A$ of participants
such that $d(A,A)=0$, $d(A,B)=d(B,A)$, and $d(A,B) \leq d(A,C) + d(C,B)$. 
Then, time information is added to the protocol. 
First, we add the time when a message was sent or received as a subindex $P_{t_1} \to V_{t_2}$.
Second, time constraints associated to the metric space are added:
(i) the sending and receiving times of a message differ by the distance between them
and
(ii) the time difference between two consecutive actions of a participant must be greater or equal to zero. 
Third, the distance bounding constraint of the verifier is represented as an arbitrary distance $d$.
Time constraints 
are written using 
quantifier-free formulas
in linear real arithmetic.
For convenience, in linear
equalities and inequalities
(with $<$, $\leq$, $>$ or $\geq$),
we 
allow both
$2*x = x + x$ and the monus function
$x \dot{-} y =
\textit{if}\ y < x\ \textit{then}\ x - y\ \textit{else}\ 0$
as definitional
extensions.

In the following timed sequence of actions, a vertical bar is included to differentiate between the process and some constraints associated to the metric space. We remove the constraint $\textit{open}(N_P,S_P,\textit{commit}(N_P,S_P))$ for simplification.\\

{\small\noindent
\[
\begin{array}{@{}r@{}l@{\;}l@{}}
P_{t_1} \rightarrow V_{t'_1} &\::\: \textit{commit}(N_P,S_P) 
& \mid t'_1 = t_1 + d(P,V)\\
V_{t_2} \rightarrow P_{t'_2} &\::\: N_V & \mid t'_2 = t_2 + d(P,V) \wedge t'_1 \geq t'_1 \\
P_{t_3} \rightarrow V_{t'_3} &\::\: N_P \oplus N_V & \mid t'_3 = t_3 + d(P,V) \wedge t_3 \geq t'_2 \\
V & \::\: t'_3\: \dot{-}\: t_2 \leq 2*d\\
P_{t_4} \rightarrow V_{t'_4} &\::\: S_P & \mid t'_4 = t_4 + d(P,V)\wedge t_4 \geq t_3 \wedge t'_4 \geq t'_3\\
P_{t_5} \rightarrow V_{t'_5} &\::\: \mathit{sign}_P(N_V ; N_P \oplus N_V) & \mid t'_5 = t_5 + d(P,V) \wedge t_5 \geq t_4 \wedge t'_5 \geq t'_4\\
\end{array}
\]%
}

The Brands-Chaum protocol is designed to defend against mafia frauds, where an honest prover is outside the neighborhood of the verifier 
(i.e., $d(P,V) > d$)
but an intruder is inside (i.e., $d(I,V) \leq d$), pretending to be the honest prover.
The following is an example of an \emph{attempted} mafia fraud,
in which 
the intruder simply forwards messages back and forth between the prover and the verifier.
We write $I(P)$ to denote an intruder pretending to be an honest prover $P$.\\

{\small\noindent
\[
\begin{array}{@{}r@{}r@{}l@{}l@{}l@{}}
P_{ t_1} &\rightarrow &I_{t_2} &: \textit{commit}(N_P,S_P) 
& \mid t_2 = t_1 + d(P,I)\\
I(P)_{t_2} &\rightarrow &V_{t_3}&:  \textit{commit}(N_P,S_P) 
& \mid t_3 = t_2 + d(V,I)\\
V_{t_3} &\rightarrow &I(P)_{t_4} &: N_V & \mid t_4 = t_3 + d(V,I)\\
I_{t_4} &\rightarrow &P_{t_5} &: N_V & \mid t_5 = t_4 + d(P,I)\\
P_{t_5} &\rightarrow &I_{t_{6}} &: N_P \oplus N_V & \mid t_{6} = t_5 + d(P,I)\\
I(P)_{t_{6}} &\rightarrow &V_{t_{7}} &: N_P \oplus N_V & \mid t_{7} = t_{6} + d(V,I)\\
& & V & : t_{7} \dot{-} t_3 \leq 2*d\\
P_{t_{8}} &\rightarrow &I_{t_{9}} &: S_P & \mid t_{9} = t_{8} + d(P,I) \wedge t_8 \geq t_5\\
I(P)_{t_{10}} &\rightarrow &V_{t_{11}} &: S_P & \mid t_{11} = t_{10} + d(V,I)  \wedge t_{11} \geq t_7\\
I(P)_{t_{12}} &\rightarrow &V_{t_{13}} &: sign_P(N_V ; N_P \oplus N_V) & \mid t_{13} = t_{12} + d(V,I)  \wedge t_{13} \geq t_{11}
\end{array}
\]%
}

\noindent
Note that, in order for this trace to be consistent with the metric space, 
it would require that $2 * d(V,I) + 2* d(P,I) \leq 2*d$, which is unsatisfiable
by $d(V,P) > d > 0$ and the triangular inequality $d(V,P) \leq d(V,I) + d(P,I)$, which implies that the attack is not possible.

However, a distance hijacking 
attack is possible (i.e., the time and distance constraints are satisfiable)
where an intruder located outside the neighborhood of the verifier  (i.e., $d(V,I) > d$) 
succeeds in convincing the verifier that he is inside the neighborhood by 
exploiting the presence of an honest prover in the neighborhood (i.e., $d(V,P) \leq d$) to achieve his goal.
The following is an example of a \emph{successful} distance hijacking, 
in which 
the intruder listens to the exchanges messages between the prover and the verifier
but builds the last  message.\\

{\small\noindent
\[
\begin{array}{@{}r@{}l@{\;}l@{}}
P_{t_1} \rightarrow V_{t_2}~~~~~ &\::\: \textit{commit}(N_P,S_P) 
& \mid t_2 = t_1 + d(P,V)\\
V_{t_2} \rightarrow P_{t_3},I_{t'_3} &\::\: N_V & \mid t_3 = t_2 + d(P,V) \wedge t'_3 = t_2 + d(I,V)\\
P_{t_3} \rightarrow V_{t_4},I_{t'_4} &\::\: N_P \oplus N_V & \mid t_4 = t_3 + d(P,V) \wedge t'_4 = t_3 + d(I,V) \\
V~~~~~~~ & \::\: t_4\: \dot{-}\: t_2 \leq 2*d\\
P_{t_5} \rightarrow V_{t_6}~~~~~ &\::\: S_P & \mid t_6 = t_5 + d(P,V)\wedge t_5 \geq t_3 \wedge t_6 \geq t_4\\
I(P)_{t_7} \rightarrow V_{t_8}~~~~~ &\::\: \mathit{sign}_I(N_V ; N_P \oplus N_V) & \mid t_8 = t_7 + d(I,V) \wedge t_7 \geq t'_4 \wedge t_8 \geq t_6\\
\end{array}
\]%
}%

\section{A Timed Process Algebra}\label{sec:mnpa-process}\label{sec:process-time}
 
In this section,
we present our timed process algebra 
and its intended semantics.  We restrict ourselves to a semantics that can be used to reason about time and distance.  We discuss how this could be extended in Section \ref{sec:conclusions}. 
To illustrate our approach, we use Maude-NPA's  process algebra and semantics described in \cite{YEMM+16}, 
extending it with a global clock and time information.

\subsection{New Syntax for Time}
\label{sec:process-time-syntax}

In our timed protocol process algebra, the behaviors of both honest principals
and the intruders are represented  by \emph{labeled processes}. 
Therefore, a protocol is specified as a set of labeled processes.
Each process performs a sequence of actions, namely sending ($+m$) or receiving ($-m$)
a message $m$, but without knowing who actually sent or received it. 
Each process may also perform deterministic or non-deterministic choices. 
We define a protocol $\caP$ in the 
timed protocol process algebra, written $\SpecTPA$,
as a pair of the form $\SpecTPA = ((\TPASymbols,\TPAPEq),\allowbreak \ProcTPA)$,
where $(\TPASymbols,\TPAPEq)$ 
is the equational theory specifying the equational properties
of the cryptographic functions and the state structure,
and $\ProcTPA$ is a $\TPASymbols$-term denoting a \emph{well-formed} timed process. 
The timed protocol process algebra's syntax $\TBNFSymbols$ is 
parameterized 
 by a sort \sort{Msg} of messages. 
 Moreover, time is represented by a new sort \sort{Real},
 since we 
allow conditional expressions on time using  
linear arithmetic for the reals.

Similar to \cite{YEMM+16}, processes support  four different kinds of choice: (i) a process expression
$P\ ?\ Q$ supports \emph{explicit non-deterministic choice} between
P and Q; (ii) a choice variable $X_{?}$ appearing in a
send message expression +m supports \emph{implicit non-deterministic
choice} of its value, which can furthermore be 
an \emph{unbounded} non-deterministic
choice if $X_{?}$ ranges over an infinite set; (iii) a conditional
\textit{if C then P else Q} supports \emph{explicit deterministic choice}
between P and Q determined by the result of its condition $C$; and
(iv) a receive message expression $-m(X_1,...,X_n)$ supports 
\emph{implicit deterministic} choice about accepting or rejecting
a received message, depending on whether or not it matches
the pattern $m(X_1,...,X_n)$.  This deterministic choice is implicit, 
but it could be made explicit by replacing $-m(X_1,...,X_n) \cdot P$
by the semantically equivalent conditional expression
$-X . \textit{ if } X = m(X_1,...,X_n) \textit{ then } P \textit{ else }\nil \cdot P$, where $X$ is a variable
of sort \sort{Msg}, which therefore accepts any message.

  The timed process algebra has the following syntax, also similar to that of \cite{YEMM+16} plus
  the addition of the suffix $@\Real$ to the sending and receiving actions:

\noindent
\begin{align}
 \ProcConf~ &::= \LProc ~|~ \ProcConf ~\&~ \ProcConf ~|~ \emptyset \notag\\[-.8ex]
 \ProcId~ &::= (\Role, \Nat)\notag\\[-.8ex]
 \LProc ~&::= (\ProcId, \Nat)~ \Proc \notag\\[-.8ex]
 \Proc~ &::=  \nil ~ |~ +(\Msg@\Real)~ | ~ -(\Msg@\Real) ~|~ \Proc \cdot \Proc  ~|  \notag \\[-.8ex]
              &~~~~~~~ \Proc~?~\Proc ~ | ~\textit{if} ~\Cond~ \textit{then} ~\Proc~ \textit{else} ~\Proc    \notag
\end{align} 

\begin{itemize}
\item $\ProcConf$ stands for a \emph{process configuration}, i.e., a set of labeled processes, where the symbol \& 
is used to denote set union for sets of labeled processes. 
\item $\ProcId$ stands for a \emph{process identifier}, where $\Role$ 
refers to the role of the process in the protocol (e.g., prover or verifier) and 
$\Nat$ is a natural number denoting the identity of the process, 
which distinguishes different instances (sessions) of a process specification.
\item $\LProc$ stands for a \emph{labeled process}, i.e., a process $\Proc$ with
 a label $(\ProcId,J)$. For convenience, we sometimes write $(\Role,I,J)$, where
$J$ indicates that the action at stage $J$ of the process $(\Role,I)$
will be the next one to be executed, i.e., the first $J-1$ actions of the process for role $\Role$ have  already been executed. 
Note that the $I$ and $J$ of a process $(\Role,I,J)$ are omitted in a protocol specification.
\item $\Proc$ defines the actions that can be executed within a process, where ${+\Msg@T}$, and ${-\Msg@T}$ respectively denote
 sending out a message or receiving a message $\Msg$.  Note that $T$ must be a variable where the underlying metric space
 determines the exact sending or receiving time, which can be used later in the process.
Moreover, ``$\Proc~\cdot~\Proc$" denotes \emph{sequential composition} of processes,
 where 
symbol \verb!_._! is associative and has the empty process $\nil$ 
as identity. Finally, 
``$\Proc~?~\Proc$" denotes an explicit \emph{nondeterministic choice}, whereas 
``$\textit{if} ~\Cond~ \textit{then} \allowbreak ~\Proc~ \textit{else}
~\Proc$" denotes an explicit \emph{deterministic choice}, whose
continuation depends on the 
satisfaction of the constraint $\Cond$. 
Note that choice is explicitly represented by either 
a non-deterministic choice between $P_1~?~P_2$
or 
by the deterministic evaluation of a conditional expression
$\textit{if} ~\Cond~ \textit{then} \allowbreak ~P_1~ \textit{else} ~P_2$,
but it is also implicitly represented by the instantiation of a variable in different runs.

\end{itemize}
  In all process specifications we assume
four disjoint kinds of variables, similar to the variables of \cite{YEMM+16} plus  time variables:
    \begin{itemize}
        	\item \textbf{\emph{fresh variables}}: each one of these variables receives  
	 a \emph{distinct constant value} from a data type \sort{V_{fresh}},
	 denoting unguessable values such as nonces. 
        	Throughout this paper we will denote this kind of variables
        	as $f,f_1,f_2,\ldots$.

        	\item \textbf{\emph{choice variables}}: variables  first
        	appearing in a \emph{sent message} $\mathit{+M}$, which can be substituted 
        	by any value arbitrarily chosen from a possibly infinite domain.
        A choice variable indicates an  \emph{implicit non-deterministic choice}.
        Given a protocol with choice variables, each possible substitution of 
        these variables denotes a possible run of the protocol.
        	We always denote choice variables by 
        	letters postfixed with the symbol ``?'' as a subscript, e.g., $\chV{A},\chV{B},\ldots$.

        	\item \textbf{\emph{pattern variables}}:  variables first appearing
        	in a \emph{received message} $\mathit{-M}$. These variables will be instantiated 
        	when matching sent and received messages. 
        	\emph{Implicit deterministic choices} are indicated by
                terms containing pattern variables, 
        	since failing to match a pattern term leads to the rejection of a message.
        	A pattern term  plays the implicit role of a guard, 
        	so that, depending on the different ways of matching it, the protocol can have different continuations.
Pattern variables are written with uppercase letters, e.g.,
		    $A,B,N_A,\ldots$.
   
\item \textbf{\emph{time variables}}: a process cannot access the global clock, which implies that  
a time variable $T$ of a reception or sending action $+(M@T)$ can never appear in $M$ but
can appear in the remaining part of the process.
Also, given a receiving action $-(M_1@t_1)$ and a sending action $+(M_2 @ t_2)$ 
in a process of the form $P_1 \cdot -(M_1@t_1) \cdot P_2 \cdot +(M_2@t_2) \cdot P_3$, 
the assumption that timed actions
are performed from left to right forces the constraint $t_1 \leq t_2$.
		    Time variables are always written with a (subscripted) $t$, e.g.,
		    $t_1,t'_1,t_2,t'_2,\ldots$.
		    
        \end{itemize} 

These conditions about variables
are formalized by 
the function
$\mathit{wf}: \mathit{\Proc} \rightarrow \mathit{Bool}$
defined in Figure~\ref{fig:wf@},
for
\emph{well-formed} processes.
The definition of $\mathit{wf}$ uses an auxiliary function 
$\mathit{shVar} : \mathit{\Proc}  \rightarrow \mathit{VarSet}$, which is defined in Figure~\ref{fig:bvar@}.

\begin{figure}[h!]
\small
\begin{align*}
& \wf(P \cdot +(M@T)) = \wf(P) \\[-.5ex]
&   ~~~~~~~~ \textit{if} ~ (\var{M} \cap   \var{P}) \subseteq \bvar{P} \wedge T\notin\var{M}\cup\var{P}\\[-.5ex]
& \wf( P \cdot -(M@T)) = \wf(P)    \\[-.5ex]
& ~~~~~~~~ \textit{if} ~ (\var{M} \cap   \var{P}) \subseteq \bvar{P}  \wedge T\notin\var{M}\cup\var{P}\\[-.5ex]
& \wf( P \cdot 
 (\textit{if} ~~T~ \textit{then} ~Q~ \textit{else} ~R))
  = \wf(P \cdot Q) \wedge \wf(P \cdot R) 
  \\
&~~~~~~~~  \textit{if} \  P \neq \nil \ \textit{and } \ Q \neq \nil \ \textit{and }  \var{T} \subseteq \bvar{P}\\
&\wf(P \cdot (Q~?~R) )= \wf(P \cdot Q)  \wedge \wf(P \cdot R) 
\ \ \textit{ if } Q \neq \nil \ \textit{or} R \neq \nil \\
&\wf(P \cdot ~\nil) = \wf(P) \\[-.5ex]
&\wf(\nil) = \textit{True}.
\end{align*}%
\vspace{-7ex}
\caption{The well-formed function}
\label{fig:wf@}
\vspace{-4ex}
\begin{align*}
& \bvar{+(M @ T) ~\cdot P} =  \var{M} \cup \bvar{P} \\[-.5ex]
& \bvar{-(M @ T)~\cdot P} =   \var{M} \cup \bvar{P}  \\[-.5ex]
& \bvar{ (\textit{if} ~T~ \textit{then} ~P~ \textit{else} ~Q) ~\cdot R} = 
    \var{T} \cup (\bvar{P} \cap \bvar{Q}) \cup \bvar{R}  \\[-.5ex]
& \bvar{ (P~?~Q)~\cdot R } =  ( \bvar{P} \cap \bvar{Q}) \cup \bvar{R} \\[-.5ex]
& \bvar{\nil} = \emptyset
\end{align*}%
\vspace{-6ex}
\caption{The shared variables auxiliary function}
\label{fig:bvar@}
\end{figure}

\begin{example}\label{ex:brands-and-chaum-process-time}
Let us specify the Brands and Chaum protocol of Example~\ref{ex:brands-and-chaum}, where variables are distinct between processes.
A nonce is represented as $n(\chV{A},f)$, whereas a secret value is represented as $s(\chV{A},f)$. 
The identifier of each process is represented by a choice variable $\chV{A}$.
Recall that there is an arbitrary distance $d > 0$. 

\noindent
{\small
 \begin{align*}
(\mathit{Verifier}):\ 
&{-}(\textit{Commit} @ t_1)\ \cdot \\[-.5ex]
&{+}(n(\chV{V},f_1) @ t_2)\ \cdot \\[-.5ex]
&{-} ((n(\chV{V},f_1) \oplus N_P)@t_3)\ \cdot \\[-.5ex]
&\textit{if}\ t_3 \dot{-} t_2 \leq 2* d \\[-.5ex]
& \textit{then}\ {-} (S_P@t_4)\ \cdot \\[-.5ex]
&\hspace{5ex}\textit{if}\ \textit{open}(N_P,S_P,\textit{Commit}) \\[-.5ex]
&\hspace{5ex} \textit{then}\ {-} (\textit{sign}(P,n(\chV{V},f_1) ; N_P \oplus n(\chV{V},f_1))@t_5) \ \textit{else}\ \nil\\
&\textit{else}\ \nil
\\
(\mathit{Prover}):\  
&{+}(\textit{commit}(n(\chV{P},f_1),s(\chV{P},f_2))@t_1)\ \cdot \\[-.5ex]
&{-}(N_V@t_2)\ \cdot \\[-.5ex]
&{+}((N_V \oplus n(\chV{P},f_1))@t_3)\ \cdot \\[-.5ex]
&{+}(s(\chV{P},f_2)@t_4)\ \cdot \\[-.5ex]
&{+}(\textit{sign}(\chV{P},N_V ; n(\chV{P},f_2) \oplus N_V)@t_5) 
\hspace{20ex} 
 \end{align*}%
}
 \end{example}
 
\subsection{Timed Intruder Model}
The active Dolev-Yao intruder model is followed, which implies an intruder can intercept, forward, or create messages
from received messages.
However, intruders are \emph{located}.  Therefore,
they cannot change the physics of the metric space, e.g., cannot send messages from a different location 
or intercept a message that it is not within range. 

In our timed intruder model,
we consider several located intruders,
modeled by the distance function $d: \ProcId \times \ProcId \to \Real$,
each with a family of capabilities (concatenation, deconcatenation, encryption, decryption, etc.),
and each capability may have arbitrarily many instances. 
The combined actions of two intruders
requires time, i.e., their distance;
but a single intruder can perform
many actions in zero time. 
Adding time cost to single-intruder
actions could be done with
additional time constraints, but is
outside the scope of this paper.
Note that, unlike in the standard Dolev-Yao model, we cannot assume just one intruder, since the time required for a principal to communicate with a given intruder is an observable characteristic of that intruder.
 Thus, although the Mafia fraud and distance hijacking attacks considered in the experiments presented in this paper only require 
configurations with just one prover, one verifier and one intruder,
the framework itself allows general participant configurations with multiple intruders. 

 \begin{example}\label{ex:dy}
 In our timed process algebra, the family of capabilities associated to an intruder $k$ are also described as processes.
 For instance, concatenating two received messages  is represented by the process (where time variables $t_1,t_2,t_3$ are not actually used by the process)
 $$
 (\mathit{k.Conc}):\ 
{-}(X @ t_1)\ \cdot 
{-} (Y@t_2)\ \cdot 
{+} (X ; Y@t_3)
$$
 
\noindent
and extracting one of them from a concatenation
is described by the process
 $$
 (\mathit{k.Deconc}):\ 
{-}(X; Y @ t_1)\ \cdot 
{+} (X@t_2)
$$
Roles of intruder capabilities include the identifier of the intruder, and it 
is possible to combine several intruder capabilities from the same or from different intruders.
For example, we may say that the ${+} (X ; Y@T)$ of a process $\mathit{I1.Conc}$ associated to an intruder $I1$
may be synchronized with the ${-}(X; Y @T')$ of a process $\mathit{I2.Deconc}$ associated to an intruder $I2$.
The metric space fixes $T' = T + d(I1,I2)$, where $d(I1,I2) > 0$ if $I1 \neq I2$ and $d(I1,I2)=0$ if $I1=I2$.

A special \emph{forwarding} intruder capability, not considered in the standard Dolev-Yao model, has to be included in order to take 
into account the time travelled by a message from an honest participant to the intruder and later to another 
participant, probably an intruder again.
 $$
 (\mathit{k.Forward}):\ 
{-}(X @ t_1)\ \cdot 
{+} (X@t_2)
$$
 \end{example}
 
\subsection{Timed Process Semantics}\label{sec:semanticsTPA}
\label{sec:process-time-semantics}

A \emph{state} of a protocol $\caP$ consists of a set 
of (possibly partially executed) \emph{labeled processes}, a set of terms in the network $\{Net\}$, and the global clock.
That is, a state is a term of the form 
$\{ LP_1 \,\&\, \cdots \,\&\, LP_n ~|~ \{\textit{Net}\} ~|~ \bar{t}\}$.
In the timed process algebra, the only time information available to a process is the variable
$T$ associated to input and output messages $M @ T$.
However, once these messages have been sent or received, we include them
in the network \textit{Net} with extra information.
When a message $M @ T$ is sent, we store 
$M\ @\ ({A : t} \to \emptyset)$ denoting that message $M$ was sent by process $A$ 
at the global time clock $t$, and propagate $T \mapsto t$ within the process $A$. 
When this message is received by an action $M' @ T'$ of process $B$ (honest participant or intruder) at the global clock time $t'$,
$M$ is matched against $M'$ modulo the cryptographic functions,
$T' \mapsto t'$ is propagated within the process $B$, 
and $B : t'$ is added to the stored message,
following the general pattern
$M\ @\ ({A : t} \to (B_1 : t_1 \cdots B_n : t_n))$.

The rewrite theory $(\TPAStateSymbols, \TPAPEq , \TPARls)$ characterizes the behavior of 
a protocol $\caP$, where  $\TPAStateSymbols$ extends $\TPASymbols$, 
by adding state constructor symbols.
We assume that a protocol run begins with an empty state, i.e., a state with an empty
set of labeled processes, an empty network, and at time zero.
Therefore, the initial empty state is always of the form 
$ \{ \emptyset ~|~ \{\emptyset\}  ~|~ 0.0\}$.
Note that, in a specific run, all the distances are provided a priori according to the metric space and a chosen topology, whereas in a symbolic analysis, they will simply be variables, probably occurring within time constraints.

State changes are defined by a set $\TPARls$ of \emph{rewrite rules}
given below. Each transition rule in $\TPARls$ is labeled with a tuple $\mathit{(ro,i, j, a,n,t)}$, where:

\begin{itemize}
	\item $\mathit{ro}$ is  the role of the labeled process being executed in the transition.
     
    \item $i$ denotes the instance of the same role 
   being executed in the transition. 
	
	\item $j$ denotes the process' step number since its beginning.

	\item $a$ is a ground term identifying the action that is being performed in the transition.
		  It has different possible values: 
		    ``$+m$'' or ``$-m$'' if the message $m$ was sent  (and added to the network) or received, respectively;
		    ``$m$'' if the message $m$ was sent  but did not increase the network,
		    ``$?$'' if the transition performs an explicit non-deterministic choice,  
		    ``$\mathit{T}$'' if the transition performs an explicit deterministic choice,
		    ``$\mathit{Time}$" when the global clock is incremented, or
		    ``$\mathit{New}$" when a new process is added.
	\item $n$ is a number that, if the action that is being executed is an explicit choice, indicates which branch has been chosen as the process continuation.
		  In this case $n$ takes the value of either $1$ or $2$.
	      If the transition does not perform any explicit choice, then $n=0$.
\item $t$ is the global clock at each transition step.
	      
\end{itemize}

Note that in the transition rules $\TPARls$ shown below, 
\textit{Net} denotes the network, represented by a set of messages of the form $M\ @\ ({A : t} \to (B_1 : t_1 \cdots B_n : t_n))$,
$P$ denotes the rest of the process being executed
and
$PS$ denotes the rest of labeled processes of the state (which can be 
the empty set $\emptyset$).

\begin{itemize}
\item 
\emph{Sending a message} is represented by the two transition rules
below, depending on whether the message $M$ is stored, \eqref{eq:tpa-output-modIK}, or just discarded, \eqref{eq:tpa-output-noModIK}.
In \eqref{eq:tpa-output-modIK}, we 
store the sent message with its sending information, $(\textit{ro},i) : \bar{t}$, and add an empty set for those who will be receiving the message in the future
$(M\sigma' @ (\textit{ro},i) : \bar{t} \to \emptyset)$.

\begin{small}
 \begin{align}
&\{ (\textit{ro},i, j)~( +M @ t \cdot P) ~\&~ PS  ~|~ \{Net\}  ~|~ \bar{t} \} \notag\\[-.5ex]
&  \longrightarrow_{(\textit{ro},i,j,+(M\sigma'),0,\bar{t})}\notag\\
& \{ (\textit{ro},i, j+1)~P \sigma' ~\&~ PS ~|~ \{ (M\sigma' @ (\textit{ro},i) : \bar{t} \to \emptyset), Net\}  ~|~ \bar{t} \} \notag\\[-.5ex]
&  \textit{ if } (M\sigma' : (\textit{ro},i) : \bar{t} \to \emptyset) \notin \textit{Net}  \notag\\[-.5ex]
&  \textit{where } \sigma  \ 
 \textit{is a ground substitution}
 \textit{ binding choice variables}
\textit{ in} \  M  
\notag\\[-.5ex]
& \hspace{6.5ex} \textit{and $\sigma'=\sigma\uplus \{t \mapsto \bar{t}\}$}
\eqname{TPA++}
 \label{eq:tpa-output-modIK}
\\[1ex]
  &\{ (\textit{ro},i, j)~( +M @ t \cdot P) ~\&~ PS  ~|~ \{Net\} ~|~ \bar{t} \} \notag\\[-.5ex]
  &  \longrightarrow_{(\textit{ro}, i, j, M\sigma',0,\bar{t})}
 \{ (\textit{ro},i, j+1)~P\sigma' ~\&~ PS ~|~ \{Net\} ~|~ \bar{t}\} \notag\\[-.5ex]
 &  \textit{where } \sigma 
 \ \textit{is a ground substitution} 
\  \textit{binding choice variables} 
\textit{ in} \  M  
\notag\\[-.5ex]
& \hspace{6.5ex} \textit{and $\sigma'=\sigma\uplus \{t \mapsto \bar{t}\}$}
\eqname{TPA+}
 \label{eq:tpa-output-noModIK}
 \end{align}
\end{small}

\item 
\emph{Receiving a message} is represented by the transition rule below.
We add the reception information to the stored message,
i.e., we replace
$(M'@((\textit{ro}',k) : t' \to AS))$
by 
$(M'@((\textit{ro}',k) : t' \to (AS \uplus (\textit{ro},i):\bar{t}))$.

\begin{small}
\begin{align}
&\{	(\textit{ro},i, j)~( -(M@t)\cdot P) ~\&~ PS \mid \{(M'@((\textit{ro}',k) : t' \to AS)), Net\} ~|~ \bar{t}\}\notag\\[-.5ex]
&	\longrightarrow_{(\textit{ro},i,j,-(M\sigma'),0,\bar{t})}\notag\\
& \{	(\textit{ro},i, j+1)~P\sigma' ~\&~ PS \mid \{ (M'@((\textit{ro}',k) : t' \to (AS \uplus (\textit{ro},i):\bar{t})), Net\}   ~|~ \bar{t} \} \notag\\[-.5ex]
&  \textsf{IF} ~ \exists\sigma: M'=_{\PEq} M\sigma, \bar{t}= t'+d((\textit{ro}',k),(\textit{ro},i)),  \eqname{TPA-} \sigma'=\sigma\uplus \{t \mapsto \hat{t}\}
	\label{eq:tpa-input}
\end{align}
\end{small}

\item 
An \emph{explicit deterministic choice} is defined as follows.
More specifically, the rule~\eqref{eq:tpa-detBranch1} describes  the
\textit{then} case, i.e., if the constraint $T$ is satisfied, 
then the process continues as $P$,
whereas 
rule~\eqref{eq:tpa-detBranch2} describes the \textit{else} case, 
that is, if the constraint $T$ is \emph{not} satisfied, the process continues as $Q$.

\begin{small}
\begin{align}
 & \{ (\textit{ro},i, j)~((\textit{if} ~T~ \textit{then} ~P~ \textit{else} ~Q) \cdot R) ~\&~ PS \mid \{Net\} \mid \bar{t}\}\notag\\[-.5ex]
 & \longrightarrow_{(\textit{ro},i, j, T,1,\bar{t})}
\{ (\textit{ro},i, j+1)~(P\cdot R) \,\&\, PS \mid \{Net\} \mid \bar{t}\}
  \textsf{IF} ~T  
 \eqname{TPAif1}
 \label{eq:tpa-detBranch1}
\\
& \{ (\textit{ro},i, j)~( (\textit{if} ~T~ \textit{then} ~P~ \textit{else} ~Q) \cdot R) ~\&~ PS \mid \{Net\} \mid \bar{t}\} \notag\\[-.5ex]
& \longrightarrow_{(\textit{ro},i,j,T,2,\bar{t})}  
\{ (\textit{ro},i, j+1)~(Q\cdot R) \,\&\, PS \mid \{Net\} \mid \bar{t}\}
  \textsf{IF} \neg T 
 \eqname{TPAif2}
  \label{eq:tpa-detBranch2}
\end{align}
\end{small}

\item 
An \emph{explicit non-deterministic choice} is defined as follows.
The process can continue either as $P$, denoted by   rule~\eqref{eq:tpa-nonDetBranch1},
or as $Q$, denoted by  rule~\eqref{eq:tpa-nonDetBranch2}. 

\begin{small}
\begin{align}
& \{ (\textit{ro},i, j)~((P~?~Q)\cdot R) ~\&~ PS \mid \{Net\} \mid \bar{t}\}\notag\\[-.5ex]
&\longrightarrow_{(\textit{ro},i,j,?,1,\bar{t})} 
\{ (\textit{ro},i, j+1)~(P \cdot R) ~\&~ PS \mid \{Net\} \mid \bar{t} \} 
\eqname{TPA?1}
\label{eq:tpa-nonDetBranch1}\\[.5ex]
& \{ (\textit{ro},i, j)~((P~?~Q)\cdot R) ~\&~ PS \mid \{Net\} \mid \bar{t}\}\notag\\[-.5ex]
&\longrightarrow_{(\textit{ro},i,j,?,2,\bar{t})} 
\{ (\textit{ro},i,j+1)(Q \cdot R) ~\&~ PS \mid \{Net\} \mid \bar{t}\} 
\eqname{TPA?2}
\label{eq:tpa-nonDetBranch2}
\end{align}
\end{small}

\item 
\emph{Global Time advancement} is represented by the transition rule below
that increments the global clock enough to make one sent message arrive to its closest destination.

{\small
\begin{align}
&\{	PS \mid \{Net\} ~|~ \bar{t}\} 
	\longrightarrow_{(\bot,\bot,\bot,\textit{Time},0,\bar{t} + t')}
\{	PS \mid \{Net\} ~|~ \bar{t}+ t'\}\notag\\[-.5ex]
&   \textsf{IF} ~ t'=\textit{mte}(PS,Net,\bar{t}) \wedge t' \neq 0 
        \eqname{PTime}
	\label{eq:time}
\end{align}
}
\noindent where the function $\mathit{mte}$ is defined as follows:\\

{\small
\noindent
\begin{itemize}
\item[] $\mathit{mte}(\emptyset,Net,\bar{t}) = \infty$
\item[] $\mathit{mte}(P \& PS,Net,\bar{t}) = \mathit{min}(\mathit{mte}(P,Net,\bar{t}),\mathit{mte}(PS,Net,\bar{t}))$
\item[] $\mathit{mte}((\textit{ro},i,j)\ \nil,Net,\bar{t}) = \infty$
\item[] $\mathit{mte}((\textit{ro},i,j)\ +(M@t) \cdot P,Net,\bar{t}) = 0$
\item[] $\mathit{mte}((\textit{ro},i,j)\ -(M@t) \cdot P,Net,\bar{t}) = \\
~~~~~ ~~~ \mathit{min}\left(\left\{\begin{array}{@{}l@{}l@{}}d((\textit{ro},i),(\textit{ro}',i'))\mid &\ (M'@(\textit{ro}',i') : t_0 \to AS)\in Net\ \\ &\wedge \exists\sigma: M\sigma =_B M' \end{array}\right\}\right)$
\item[] $\mathit{mte}((\textit{ro},i,j)\ (\textit{if} ~T~ \textit{then} ~P~ \textit{else} ~Q) \cdot R,Net,\bar{t}) = 0$
\item[] $\mathit{mte}((\textit{ro},i,j)\ P_1 ? P_2,Net,\bar{t}) = 0$\\
\end{itemize}
}

\noindent
Note that the function $\mathit{mte}$ evaluates to $0$ if some instantaneous action by the previous rules can be performed. Otherwise, $\mathit{mte}$ computes 
the smallest non-zero
time increment required for some already sent message (existing in the network) to be received by some process
(by matching with such an existing message in the network).

\paragraph{\bf Remark}
The timed process semantics assumes a metric space with a distance function $d: \ProcId \times \ProcId \to \Real$ 
such that (i) $d(A,A)=0$, (ii) $d(A,B)=d(B,A)$, and (iii) $d(A,B) \leq d(A,C) + d(C,B)$. 
For every message $M\ @\ ({A : t} \to (B_1 : t_1 \cdots B_n : t_n))$
stored in the network \textit{Net}, 
 our semantics assumes that (iv) $t_i = t + d(A,B_i)$, $\forall 1\leq i\leq n$.
Furthermore, 
according to our wireless communication model, 
our semantics assumes (v) a \emph{time sequence monotonicity} property,
i.e., 
there is no other process $C$ such that $d(A,C) \leq d(A,B_i)$ for some $i$, $1\leq i\leq n$, 
and $C$ is not included in the set of recipients of the message $M$.
Also, for each class of attacks such as the Mafia fraud or the hijacking attack, 
(vi) some extra topology constraints may be necessary.
However, in Section~\ref{sec:simple}, timed processes are transformed into untimed processes with time constraints
and the transformation takes care only of conditions (i), (ii), and (iv).
For a fixed number of participants,
all the instances of the triangle inequality (iii) 
as well as constraints (vi)
should be added by the user.
In the general case, conditions (iii), (v), and (vi) can  be partially specified 
and fully checked on a successful trace; 
see Definition~\ref{def:realizable} in the additional supporting material. \\

\item 
New processes can be added as follows.

\vspace{-1.5ex}
{\small
 \begin{align}
 \left \{
 \begin{array}{@{}l@{}}
  \forall \  (ro)~ P_k \in \ProcPA\notag\\[.5ex]
  \{ PS \mid \{Net\} \mid \bar{t}\} \notag\ \\[-.5ex]
  \longrightarrow_{(\textit{ro}, i + 1, 1,New,0,\bar{t})}\\
 \{ (\textit{ro}, i+1, 1,\chV{x}\sigma,\chV{y}\sigma)~P_k\sigma\rho_{ro,i+1} ~\&~ PS \mid \{Net\}  \mid \bar{t}\}\notag\\ [1ex]
  \textit{where } \rho_\mathit{ro,i+1} 
 \ \textit{is a fresh substitution}, \notag\\[-.5ex]
\ \ \ \sigma  \ 
 \textit{is a ground substitution}
 \textit{ binding } 
\chV{x} \textit{ and }\chV{y}, \notag
\mbox{ and }
i= \MaxProcId(PS, ro)  
 \end{array}
 \right \}\eqname{TPA\&}
 \label{eq:tpa-new}
 \end{align}
}
 \vspace{-1.5ex}

The auxiliary function $\MaxProcId$ counts the instances of a role

\begin{small}
\begin{itemize}
\item[] $\MaxProcId(\emptyset, ro) = 0 $
\item[]
 $\MaxProcId((\textit{ro}', i, j) P \& PS, ro)  =  
 \left\{\begin{array}{@{}l@{}c@{}}
 max(\MaxProcId(PS, ro), i) & \textit{ if } \textit{ro} = \textit{ro}'\\
 \MaxProcId(PS, ro) & \textit{ if } \textit{ro} \neq \textit{ro}'
 \end{array}\right.
 $
\end{itemize}
\end{small}

\noindent
where $PS$ denotes a process configuration, $P$ a process, and $\textit{ro}, \textit{ro}'$  role names.
\end{itemize}

 Therefore, the behavior of a timed protocol in the process algebra is
 defined by the set of transition rules
 $\TPARls = \{ \eqref{eq:tpa-output-modIK},\allowbreak \eqref{eq:tpa-output-noModIK},\allowbreak 
              \eqref{eq:time},\allowbreak
              \eqref{eq:tpa-input},\allowbreak
              \eqref{eq:tpa-detBranch1},\allowbreak  \eqref{eq:tpa-detBranch2},\allowbreak 
              \eqref{eq:tpa-nonDetBranch1},\allowbreak  \eqref{eq:tpa-nonDetBranch2} \} \cup
              \eqref{eq:tpa-new} $.
              
\begin{figure*}[h!]
	\vspace{-.5cm}
	\centering
	\includegraphics[width=.8\linewidth]{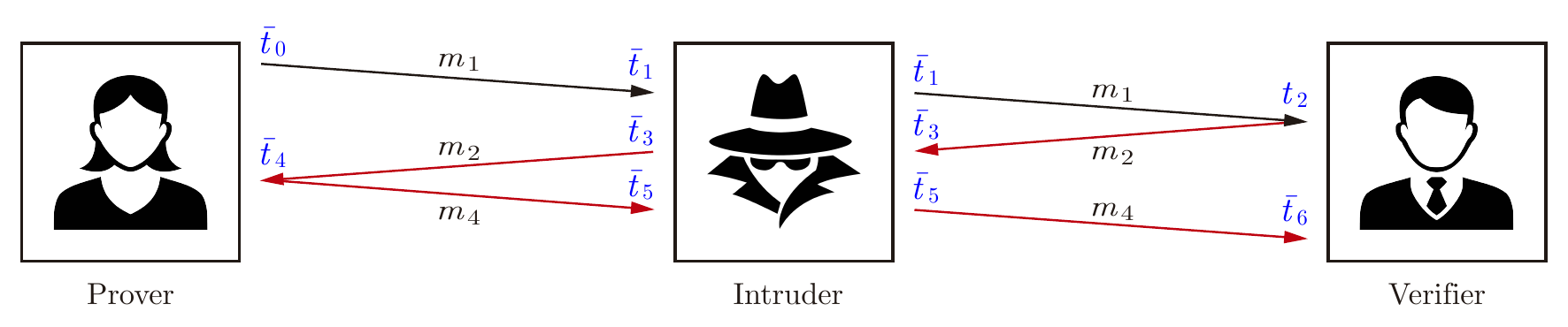}
	\vspace{-.3cm}
{\scriptsize
\begin{align*}
\{ \emptyset{\mid} \{\emptyset\} {\mid} 0.0\} 
\rightarrow_{p,0,1,{\textit New}}
&\{ (p,0,1): {+}(m_1@t_1) \cdots ~|~ \{\emptyset\}  ~|~ \textcolor{blue}{\bar{t}_0}=0.0\}
\ \ 
 m_1 = \textit{commit}(n(p,f_1),s(p,f_2))\\
\rightarrow_{i,0,1,{\textit New}}
&\left\{ 
\begin{array}{@{}l@{\ }}
(p,0,1): {+}(m_1@t_1) \cdots \\
(i.F,0,1): {-}(X@ t'_1) \cdot  {+}(X@ t'_2)
\end{array}
\right| 
\{\emptyset\}  ~|~ \textcolor{blue}{\bar{t}_0}\}\\
\rightarrow_{p,0,1,+(m_1)}
&\left\{ 
\begin{array}{@{}l@{\ }}
(p,0,2): {-}(N_V@t_2)\ \cdots \\
(i.F,0,1): {-}(X@ t'_1) \cdot  {+}(X@ t'_2)
\end{array}
\right| 
\{(m_1@(p,0):\textcolor{blue}{\bar{t}_0} \to \emptyset)\}  ~|~ \textcolor{blue}{\bar{t}_0}\}\\
\rightarrow_{\textit Time}
&\left\{ 
\begin{array}{@{}l@{\ }}
(p,0,2): {-}(N_V@t_2)\ \cdots \\
(i.F,0,1): {-}(X@ t'_1) \cdot  {+}(X@ t'_2)
\end{array}
\right| 
\{(m_1@(p,0):\textcolor{blue}{\bar{t}_0} \to \emptyset)\}  ~|~ \textcolor{blue}{\bar{t}_1}=1.0\}
\\
\rightarrow_{i,0,1,-(m_1)}
&\left\{ 
\begin{array}{@{}l@{\ }}
(p,0,2): {-}(N_V@t_2)\ \cdots \\
(i.F,0,2): {+}(m_1@ t'_2)
\end{array}
\right| 
\{(m_1@(p,0):\textcolor{blue}{\bar{t}_0} \to (i.F,0):\textcolor{blue}{\bar{t}_1})\}  ~|~ \textcolor{blue}{\bar{t}_1}\}\\
\rightarrow_{i,0,3,+(m_1)}
&\left\{ 
\begin{array}{@{}l@{\ }}
(p,0,2): {-}(N_V@t_2)\ \cdots \\
(i.F,0,2): \nil
\end{array}
\right| 
\left\{
\begin{array}{@{}l@{}}
(m_1@(i.F,0):\textcolor{blue}{\bar{t}_1} \to \emptyset
\end{array}
\right\}
  ~|~ \textcolor{blue}{\bar{t}_1} \}\\
\rightarrow_{v,0,1,{\textit New}}
&\left\{ 
\begin{array}{@{}l@{\ }}
(p,0,2): {-}(N_V@t_2)\ \cdots \\
(v,0,1): {-}(\textit{Commit} @ t''_1) \cdots 
\end{array}
\right| 
\left\{
\begin{array}{@{}l@{}}
(m_1@(i.F,0):\textcolor{blue}{\bar{t}_1} \to \emptyset)
\end{array}
\right\}
  ~|~ \textcolor{blue}{\bar{t}_1} \}\\
\rightarrow_{\textit Time}
&\left\{ 
\begin{array}{@{}l@{\ }}
(p,0,2): {-}(N_V@t_2)\ \cdots \\
(v,0,1): {-}(\textit{Commit} @ t''_1) \cdots 
\end{array}
\right| 
\left\{
\begin{array}{@{}l@{}}
(m_1@(i.F,0):\textcolor{blue}{\bar{t}_1} \to \emptyset)
\end{array}
\right\}
  ~|~ \textcolor{blue}{\bar{t}_2}=2.0 \}\\
\rightarrow_{v,0,1,-(m_1)}
&\left\{ 
\begin{array}{@{}l@{\ }}
(p,0,2): {-}(N_V@t_2)\ \cdots \\
(v,0,2): {+}(m_2 @ t''_2) \cdots 
\end{array}
\right| 
\left\{
\begin{array}{@{}l@{}}
(m_1@(i.F,0):\textcolor{blue}{\bar{t}_1} \to (v,0):\textcolor{blue}{\bar{t}_2})
\end{array}
\right\}
  ~|~ \textcolor{blue}{\bar{t}_2} \}
\ \ 
   m_2=n(v,f_3)\\
\rightarrow_{v,0,2,+(m_2)}
&\left\{ 
\begin{array}{@{}l@{\ }}
(p,0,2): {-}(N_V@t_2)\ \cdots \\
(v,0,3): {-} (m_3@t''_3) \cdots 
\end{array}
\right| 
\left\{
\begin{array}{@{}l@{}}
(m_2@(v,0):\textcolor{blue}{\bar{t}_2} \to \emptyset)
\end{array}
\right\}
  ~|~ \textcolor{blue}{\bar{t}_2} \}
\ \ 
   m_3=(m_2 \oplus N_P)\\
\rightarrow_{i,1,1,{\textit New}}
&\left\{ 
\begin{array}{@{}l@{\ }}
(p,0,2): {-}(N_V@t_2)\ \cdots \\
(v,0,3): {-} (m_3@t''_3) \cdots \\
(i.F,1,1): {-}(X'@ t'''_1) \cdot  {+}(X'@ t'''_2)
\end{array}
\right| 
\left\{
\begin{array}{@{}l@{}}
(m_2@(v,0):\textcolor{blue}{\bar{t}_2} \to \emptyset)
\end{array}
\right\}
  ~|~ \textcolor{blue}{\bar{t}_2} \}\\
\rightarrow_{\textit Time}
&\left\{ 
\begin{array}{@{}l@{\ }}
(p,0,2): {-}(N_V@t_2)\ \cdots \\
(v,0,3): {-} (m_3@t''_3) \cdots \\
(i.F,1,1): {-}(X'@ t'''_1) \cdot  {+}(X'@ t'''_2)
\end{array}
\right| 
\left\{
\begin{array}{@{}l@{}}
(m_2@(v,0):\textcolor{blue}{\bar{t}_2} \to \emptyset)
\end{array}
\right\}
  ~|~ \textcolor{blue}{\bar{t}_3}=3.0 \}\\
\rightarrow_{i,1,1,-(m_2)}
&\left\{ 
\begin{array}{@{}l@{\ }}
(p,0,2): {-}(N_V@t_2)\ \cdots \\
(v,0,3): {-} (m_3@t''_3) \cdots \\
(i.F,1,2): {+}(m_2@ t'''_2)
\end{array}
\right| 
\left\{
\begin{array}{@{}l@{}}
(m_2@(v,0):\textcolor{blue}{\bar{t}_2} \to (i.F,1):\textcolor{blue}{\bar{t}_3})
\end{array}
\right\}
  ~|~ \textcolor{blue}{\bar{t}_3} \}\\
\rightarrow_{i,1,2,+(m_2)}
&\left\{ 
\begin{array}{@{}l@{\ }}
(p,0,2): {-}(N_V@t_2)\ \cdots \\
(v,0,3): {-} (m_3@t''_3) \cdots \\
(i.F,1,2): \nil
\end{array}
\right| 
\left\{
\begin{array}{@{}l@{}}
(m_2@(i.F,1):\textcolor{blue}{\bar{t}_3} \to \emptyset)
\end{array}
\right\}
  ~|~ \textcolor{blue}{\bar{t}_3} \}\\
\rightarrow_{Time}
&\left\{ 
\begin{array}{@{}l@{\ }}
(p,0,2): {-}(N_V@t_2)\ \cdots \\
(v,0,3): {-} (m_3@t''_3) \cdots 
\end{array}
\right| 
\left\{
\begin{array}{@{}l@{}}
(m_2@(i.F,1):\textcolor{blue}{\bar{t}_3} \to \emptyset)
\end{array}
\right\}
  ~|~ \textcolor{blue}{\bar{t}_4}=4.0 \}\\
\rightarrow_{p,0,2,-(m_2)}
&\left\{ 
\begin{array}{@{}l@{\ }}
(p,0,3): {+}(m_4@t_3)\ \cdots \\
(v,0,3): {-} (m_3@t''_3) \cdots 
\end{array}
\right| 
\left\{
\begin{array}{@{}l@{}}
(m_2@(i.F,1):\textcolor{blue}{\bar{t}_3} \to (p,0):\textcolor{blue}{\bar{t}_4})
\end{array}
\right\}
  ~|~ \textcolor{blue}{\bar{t}_4} \}
\ \ 
 m_4=(m_2 \oplus n(p,f_1))\\
\rightarrow_{p,0,3,+(m_4)}
&\left\{ 
\begin{array}{@{}l@{\ }}
(p,0,4): {+}(m_5@t_4)\ \cdots \\
(v,0,3): {-} (m_3@t''_3) \cdots 
\end{array}
\right| 
\left\{
\begin{array}{@{}l@{}}
(m_4@(p,0):\textcolor{blue}{\bar{t}_4} \to \emptyset)
\end{array}
\right\}
  ~|~ \textcolor{blue}{\bar{t}_4} \}
\ \ 
 m_5=s(p,f_2)\\
 \rightarrow_{i,2,1,{\textit New}}
&\left\{ 
\begin{array}{@{}l@{\ }}
(v,0,3): {-} (m_3@t''_3) \cdots \\
(i.F,2,1): {-}(X''@ t''''_1) \cdot  {+}(X''@ t''''_2)
\end{array}
\right| 
\left\{
\begin{array}{@{}l@{}}
(m_4@(p,0):\textcolor{blue}{\bar{t}_4} \to \emptyset)
\end{array}
\right\}
  ~|~ \textcolor{blue}{\bar{t}_4} \}\\
 \rightarrow_{\textit Time}
&\left\{ 
\begin{array}{@{}l@{\ }}
(v,0,3): {-} (m_3@t''_3) \cdots \\
(i.F,2,1): {-}(X''@ t''''_1) \cdot  {+}(X''@ t''''_2)
\end{array}
\right| 
\left\{
\begin{array}{@{}l@{}}
(m_4@(p,0):\textcolor{blue}{\bar{t}_4} \to \emptyset)
\end{array}
\right\}
  ~|~ \textcolor{blue}{\bar{t}_5}=5.0 \}\\
\rightarrow_{i,3,1,-(m_4)}
&\left\{ 
\begin{array}{@{}l@{\ }}
(v,0,3): {-} (m_3@t''_3) \cdots \\
(i.F,2,1): {+}(m_4@ t''''_2)
\end{array}
\right| 
\left\{
\begin{array}{@{}l@{}}
(m_4@(p,0):\textcolor{blue}{\bar{t}_4} \to (i.F,2):\textcolor{blue}{\bar{t}_5})
\end{array}
\right\}
  ~|~ \textcolor{blue}{\bar{t}_5} \}\\
\rightarrow_{i,3,2,+(m_4)}
&\left\{ 
\begin{array}{@{}l@{\ }}
(v,0,3): {-} (m_3@t''_3) \cdots \\
(i.F,2,1): \ nil
\end{array}
\right| 
\left\{
\begin{array}{@{}l@{}}
(m_4@(i.F,2):\textcolor{blue}{\bar{t}_5} \to \emptyset)
\end{array}
\right\}
  ~|~ \textcolor{blue}{\bar{t}_5} \}\\
\rightarrow_{\textit Time}
&\left\{ 
\begin{array}{@{}l@{\ }}
(v,0,3): {-} (m_3@t''_3) \cdots 
\end{array}
\right| 
\left\{
\begin{array}{@{}l@{}}
(m_4@(i.F,2):\textcolor{blue}{\bar{t}_5} \to \emptyset)
\end{array}
\right\}
  ~|~ \textcolor{blue}{\bar{t}_6}=6.0 \}\\
\rightarrow_{v,0,3,-(m_4)}
&\left\{ 
\begin{array}{@{}l@{\ }}
(v,0,4): {-} (S_P@t_4) \cdots 
\end{array}
\right| 
\left\{
\begin{array}{@{}l@{}}
(m_4@(i.F,2):\textcolor{blue}{\bar{t}_5} \to (v,0):\textcolor{blue}{\bar{t}_6})
\end{array}
\right\}
  ~|~ \textcolor{blue}{\bar{t}_6}\}
\end{align*}
}
\vspace{-6ex}
\caption{Brand and Chaum execution for a prover, an intruder, and a verifier}
\label{fig:brands-and-chaum-simulation}
\vspace{-5ex}
 \end{figure*}

\begin{example}\label{ex:brands-and-chaum-simulation}
Continuing Example~\ref{ex:brands-and-chaum-process-time},
a possible run of the protocol is represented in Figure~\ref{fig:brands-and-chaum-simulation}
for a prover $p$, an intruder $i$, and a verifier $v$.
A simpler, graphical representation of the same run is included at the top of the figure. 
There, the neighborhood distance is $d= 1.0$, the distance between the prover and the verifier is $d(p,v)=2.0$, but the distance between the prover and the intruder as well as the distance between the verifier
and the intruder are $d(v,i)=d(p,i)=1.0$,
i.e., 
the honest prover $p$ is outside $v$'s neighborhood, $d(v,p) > d$, where $d(v,p) = d(v,i) + d(p,i)$.
Only the first part of the rapid message exchange sequence is represented
and the forwarding action of the intruder is denoted by $i.F$.

The prover sends the commitment $m_1 = \textit{commit}(n(p, f_1), s(p, f_2))$ at instant $\bar{t}_0=0.0$
and is received by the intruder at instant $\bar{t}_1=1.0$.
The intruder forwards $m_1$ at instant $\bar{t}_1$ and is received by the verifier at instant $\bar{t}_2=2.0$.
Then, the verifier sends $m_2 =n(v,f_3)$ at instant $\bar{t}_2$, which is received by the intruder at instant $\bar{t}_3=3.0$.
The intruder forwards $m_2$ at instant $\bar{t}_3$, which is received by the prover at instant $\bar{t}_4=4.0$.
Then, the prover sends $m_4 =(m_2 \oplus n(p,f_1))$ at instant $\bar{t}_4$ and is received by the intruder at instant $\bar{t}_5=5.0$.
Finally, the intruder forwards $m_4$ at instant $\bar{t}_5$ and is received by the verifier at instant $\bar{t}_6=6.0$.
Thus, the verifier sent $m_2$ at time $\bar{t}_2=2.0$ and received $m_4$ at time $\bar{t}_6=6.0$.
But 
the protocol cannot complete the run, since 
$\bar{t}_6-\bar{t}_2=4.0 < 2*d=2.0$ is unsatisfiable.
\end{example}
 
Our time protocol semantics can already be implemented straightforwardly
as a simulation tool.
For instance, 
\cite{MPPC+07} describes distance bounding protocols using an authentication logic, which describes the evolution of the protocol, \cite{NTU16} provides a strand-based framework for distance bounding protocols based on simulation with time constraints, and \cite{DD19}
defines distance bounding protocol using some applied-pi calculus.
Note, however, that, since the number
of metric space configurations is infinite, model checking a protocol for a \emph{concrete configuration} with a simulation tool is very limited, since it cannot prove the \emph{absence} of
an attack for \emph{all} configurations.
For this reason, we follow a symbolic
approach that can explore \emph{all}
relevant configurations.

In the following section, we provide a sound and complete protocol transformation from our timed process algebra
to  the untimed process algebra of the Maude-NPA tool. 
In order to do this, we make use of an approach introduced by Nigam et al.~\cite{NTU16} in which properties of time, which can include both those following from physics and those checked by principals, are represented by linear constraints on the reals.  As a path is built, an SMT solver can be used to check that the constraints are satisfiable, as is done in 
 \cite{NTU18}.

\section{Timed Process Algebra into Untimed Process Algebra with Time Variables and Timing Constraints}\label{sec:simple}

In this section, we consider a more general constraint satisfiability approach,
where all possible (not only some) runs are symbolically analyzed. 
This provides both a trace-based insecure statement, i.e., 
a run leading to an insecure secrecy or authentication property is discovered
given enough resources,
and an unsatisfiability-based secure statement, i.e.,
there is no run leading to an insecure secrecy or authentication property
due to time constraint unsatisfiability.

\begin{example}\label{ex:brands-and-chaum-simulation-untimed}
Consider again the run of the Brands-Chaum  protocol given in Figure~\ref{fig:brands-and-chaum-simulation}.
All the terms of sort \sort{Real}, written in blue color, are indeed variables that get an assignment during the run
based on the distance function.
Then, it is possible to obtain a symbolic trace from the run of Figure~\ref{fig:brands-and-chaum-simulation}, where the following time constraints are accumulated:
\begin{itemize}
\item[]\quad\quad $\bar{t}_1 = \bar{t}_0 + d((p,0),(i.F,0))$, $d((p,0),(i.F,0)) \geq 0$
\item[]\quad\quad $\bar{t}_2 = \bar{t}_1 + d((v,0),(i.F,0))$, $d((v,0),(i.F,0)) \geq 0$
\item[]\quad\quad $\bar{t}_3 = \bar{t}_2 + d((v,0),(i.F,1))$, $d((v,0),(i.F,1)) \geq 0$
\item[]\quad\quad $\bar{t}_4 = \bar{t}_3 + d((p,0),(i.F,1))$, $d((p,0),(i.F,1)) \geq 0$
\item[]\quad\quad $\bar{t}_5 = \bar{t}_4 + d((p,0),(i.F,2))$, $d((p,0),(i.F,2)) \geq 0$ 
\item[]\quad\quad $\bar{t}_6 = \bar{t}_5 + d((v,0),(i.F,2))$, $d((v,0),(i.F,2)) \geq 0$
\end{itemize}

Note that these constraints are unsatisfiable when 
combined with 
(i) the assumption $d > 0$,
(ii) the verifier check $\bar{t}_6-\bar{t}_2 \leq 2*d$,
(iii) the assumption that the honest prover is outside the verifier's neighborhood,
$d((p,0),\allowbreak(v,0)) > d$, 
(iv) the triangular inequality from the metric space,
$d((p,0),(v,0)) \leq d((p,0),\allowbreak(i.F,0)) + d(\allowbreak(i.F,0),(v,0))$,
and
(v) the assumption that there is only one intruder
$d((i.F,0),\allowbreak(i.F,1))=0$ and $d((i.F,0),\allowbreak(i.F,2))=0$.
\end{example}

As explained previously in the remark, 
there are some implicit conditions based on the \textit{mte} function to calculate the time increment to the closest destination of a message.
However, the \textit{mte} function disappears in the untimed process algebra and those implicit conditions are incorporated into the symbolic run.
In the following, we define a transformation of the timed process algebra 
by (i) removing the global clock; 
(ii) adding the time data into untimed messages of a process algebra without time 
(as done in \cite{NTU16});
and 
(iii) adding 
linear arithmetic conditions over the reals for the time constraints
(as is done in \cite{NTU18}).
The soundness and completeness proof of the transformation
is included in the additional supporting material at the end of the paper.

Since all the relevant time information is actually stored in 
messages of the form
$M\ @\ ({A : t} \to (B_1 : t_1 \cdots B_n : t_n))$
and controlled by the transition rules 
\eqref{eq:tpa-output-modIK},\allowbreak \eqref{eq:tpa-output-noModIK},\allowbreak 
\ and 
              \eqref{eq:tpa-input},
the mapping $\tpatopa$ of Definition~\ref{def:transf} below transforms
each message $M@t$ of a timed process
into 
a message $M\ @\ ({A : \chV{t}} \to \chV{AS})$ of an untimed process.
That is, we use 
a timed choice variable $\chV{t}$ for the 
sending time
and
a variable $\chV{AS}$ for the reception information $(B_1 : t'_1 \cdots B_n : t'_n)$ associated to the sent message.
Since choice variables are replaced by specific values, 
both $\chV{t}$ and $\chV{AS}$ 
will be replaced by the appropriate values that make the execution and all its time constraints possible.
Note that these two choice variables will be replaced by logical variables during the symbolic execution.

\begin{definition}[Adding Time Variables and Time Constraints to Untimed Processes]\label{def:transf}
The mapping $\tpatopa$ from timed processes into untimed processes and its auxiliary mapping $\tpatopax$ are defined as follows:

{\small\noindent
\begin{align*}
& \tpatopa(\emptyset) =  \mathit{\emptyset}\\[-.5ex]
& \tpatopa((\mathit{ro}{,}i{,} j)\, P\ \& \ \mathit{PS}) = (\textit{ro}{,}i{,} j)\,\tpatopax(P{,}\textit{ro}{,}i) \ \& \ \tpatopa(\mathit{PS})\\[1ex]
&  \tpatopax( \nil,\textit{ro},i) =  \nil \\[-.5ex]
&  \tpatopax(\ +(M@t) \ . \  P, \textit{ro},i) =  +(M@((\textit{ro},i) : \chV{t} \to \chV{AS})) \ .\  \tpatopax(P\gamma, \textit{ro},i) \\[-.5ex]
&\hspace*{5mm} 
 \mbox{where }\gamma=\{t\mapsto \chV{t}\}\\[-.5ex]
&  \tpatopax(\  -(M@t) \ . \  P, \textit{ro},i) =  \\[-.5ex]
&\hspace*{2mm}  
-(M@((\textit{ro}',i') : t' \to ((\textit{ro},i):t)\uplus AS)) \ .\ \\[-.5ex]
&\hspace*{2mm} 
\textit{ if } t = t'+d((\textit{ro},i),(\textit{ro}',i'))\wedge d((\textit{ro},i),(\textit{ro}',i'))\geq 0  \textit{ then }  \tpatopax(P, \textit{ro},i) \textit{ else } \nil    \\[-.5ex]
&  \tpatopax(\textit{ (if } C \textit{ then }  P \textit{ else } Q ) \  . \ R{,} \textit{ro}{,}i{,}x{,}y) \\[-.5ex]
& \hspace*{2mm} = \textit{ (if } C  \textit{ then }  \tpatopax(P{,}\textit{ro}{,}i{,}x{,}y) \textit{ else } \tpatopax(Q{,}\textit{ro}{,}i{,}x{,}y)) \  .\ \tpatopax(R{,} \textit{ro}{,}i{,}x{,}y) \\[-.5ex]
&  \tpatopax( \ (P \ ? \ Q)\ . \ R{,} \textit{ro}{,}i{,}x{,}y) \\[-.5ex]
&  \hspace*{2mm}= \ (\tpatopax(P{,}\textit{ro}{,}i{,}x{,}y) \ ? \ \tpatopax(Q{,}\textit{ro}{,}i{,}x{,}y))\ . \ \tpatopax(R{,} \textit{ro}{,}i{,}x{,}y)
\end{align*}%
}%
\vspace{-2ex}

\noindent
where 
$\chV{t}$ and $\chV{AS}$ are  choice variables different for each one of the sending actions,
$\textit{ro}',i',t',d,AS$ are pattern variables different for each one of the receiving actions,
$P$, $Q$, and $R$ are processes,
$M$ is a message, 
and
$C$ is a constraint.
 \end{definition}

\begin{example}\label{ex:brands-and-chaum-process-time-trans}
The timed processes 
of Example~\ref{ex:brands-and-chaum-process-time}
are transformed into the following untimed processes.
We remove the ``else $\nil$" branches for clarity.

\noindent
{\small
 \begin{align*}
(\mathit{Verifier}):\ 
&{-}(\textit{Commit}\ @\ A_1 : t'_1 \to \chV{V} : t_1\uplus AS_1)\ \cdot \\[-.5ex]
&\textit{if}\ t_1 = t'_1 + d(A_1,\chV{V}) \wedge d(A_1,\chV{V})\geq 0 \textit{ then} \\[-.5ex]
&{+}(n(\chV{V},f_1)\ @\ \chV{V} : \chV{t_2} \to \chV{AS_2})\ \cdot \\[-.5ex]
&{-} ((n(\chV{V},f_1) \oplus N_P)\ @\ A_3 : t'_3 \to \chV{V} : t_3\uplus AS_3)\ \cdot \\[-.5ex]
&\textit{if}\ t_3 = t'_3 + d(A_3,\chV{V}) \wedge d(A_3,\chV{V})\geq 0 \textit{ then} \\[-.5ex]
&\textit{if}\ t_3 \dot{-} \chV{t_2} \leq 2*d \textit{ then} \\[-.5ex]
&{-} (S_P\ @\ A_4: t'_4 \to \chV{V} : t_4\uplus AS_4)\ \cdot \\[-.5ex]
&\textit{if}\ t_4 = t'_4 + d(A_4,\chV{V}) \wedge d(A_4,\chV{V})\geq 0 \textit{ then} \\[-.5ex]
&\textit{if}\ \textit{open}(N_P,S_P,\textit{Commit}) \textit{ then} \\[-.5ex]
& {-} (\textit{sign}(P,n(\chV{V},f_1) ; N_P \oplus n(\chV{V},f_1))\ @\ A_5: t'_5 \to \chV{V} : t_5\uplus AS_5) \\[-.5ex]
&\textit{if}\ t_5 = t'_5 + d(A_5,\chV{V}) \wedge d(A_5,\chV{V})\geq 0
\\
(\mathit{Prover}):\  
&{+}(\textit{commit}(n(\chV{P},f_1),s(\chV{P},f_2))@\chV{P} : \chV{t_1} \to \chV{AS_1})\ \cdot \\[-.5ex]
&{-}(V;N_V@ \ A_2: t'_2 \to \chV{V} : t_2\uplus AS_2)\ \cdot \\[-.5ex]
&\textit{if}\ t_2 = t'_2 + d(A_2,\chV{P}) \wedge d(A_2,\chV{P})\geq 0 \textit{ then} \\[-.5ex]
&{+}((N_V \oplus n(\chV{P},f_1))@\chV{P} : \chV{t_3} \to \chV{AS_3})\ \cdot \\[-.5ex]
&{+}(s(\chV{P},f_2)@\chV{P} : \chV{t_4} \to \chV{AS_4})\ \cdot \\[-.5ex]
&{+}(\textit{sign}(\chV{P},N_V ; n(\chV{P},f_2) \oplus N_V)@\chV{P} : \chV{t_5} \to \chV{AS_5})) 
\end{align*}%
}%
 \end{example}

 \begin{example}\label{ex:dy-untimed}
The timed processes 
of Example~\ref{ex:dy} for the intruder 
are transformed into the following untimed processes.
Note that we use the intruder identifier $I$ 
associated to each role
instead of a choice variable $\chV{I}$.

\noindent
{\small
 \begin{align*}
 (\mathit{I.Conc}):\ 
&{-}(X @\ A_1 : t_1 \to I : t'_1\uplus AS_1)\ \cdot\\[-.5ex] 
&\textit{if}\ t'_1 = t_1 + d(A_1,I) \wedge d(A_1,I)\geq 0 \textit{ then} \\[-.5ex]
&{-}(Y @\ A_2 : t_2 \to I : t'_2\uplus AS_2)\ \cdot\\[-.5ex] 
&\textit{if}\ t'_2 = t_2 + d(A_2,I) \wedge d(A_2,I)\geq 0 \textit{ then} \\[-.5ex]
&{+} (X ; Y@I : \chV{t_3} \to \chV{AS})
\end{align*}
}
 
\noindent
{\small
 \begin{align*}
 (\mathit{I.Deconc}):\ 
&{-}(X;Y @\ A_1 : t_1 \to I : t'_1\uplus AS_1)\ \cdot\\[-.5ex] 
&\textit{if}\ t'_1 = t_1 + d(A_1,I) \wedge d(A_1,I)\geq 0 \textit{ then} \\[-.5ex]
&{+} (X@I : \chV{t_2} \to \chV{AS})
\\
 (\mathit{I.Forward}):\ 
&{-}(X @\ A_1 : t_1 \to I : t'_1\uplus AS_1)\ \cdot\\[-.5ex] 
&\textit{if}\ t'_1 = t_1 + d(A_1,I) \wedge d(A_1,I)\geq 0 \textit{ then} \\[-.5ex]
&{+} (X@I : \chV{t_2} \to \chV{AS})
\end{align*}
}

 \end{example}

Once a timed process is transformed into an untimed process with time variables and time constraints using the notation of Maude-NPA,
we rely on 
both
a soundness and completeness proof from the Maude-NPA process notation into Maude-NPA forward rewriting semantics
and on 
a soundness and completeness proof from Maude-NPA forward rewriting semantics into Maude-NPA backwards symbolic semantics,
see \cite{YEMM+16,YEMM19}.
Since the Maude-NPA backwards symbolic semantics already considers constraints in a very general sense \cite{EMMS15}, 
we only need to perform the 
additional satisfiability check for
linear arithmetic  over the reals.

\section{Timed Process Algebra into Strands in Maude-NPA }\label{sec:byhand}

This section is provided to help in understanding the experimental output.  Although Maude-NPA accepts protocol specifications in either the process algebra language or the strand space language, it still gives outputs only in the strand space notation. 
Thus, in order to make our experimental output easier to understand, we describe the translation from timed process
into strands with time variables and time constraints. 
This translation is also sound and complete,
as it imitates the transformation of Section~\ref{sec:simple}
and the transformation of 
\cite{YEMM+16,YEMM19}.

Strands~\cite{THG99} are used in Maude-NPA to represent both the actions of honest principals (with a strand 
specified for each protocol role) 
and those of an intruder 
(with a strand 
for each action an intruder is able to perform on messages).
In Maude-NPA, strands evolve over time.
The symbol $|$ is used to divide past and future.
That is, given a strand 
$[\ \textit{msg}_1^{\pm},\ \ldots,\ \textit{msg}_i^{\pm}\ |\allowbreak\ \textit{msg}_{i+1}^{\pm},\ \ldots,\ \textit{msg}_k^{\pm}\ ]$,
 messages $\textit{msg}_1^\pm,\linebreak[2] \ldots, \textit{msg}_{i}^\pm$ are
the \emph{past messages}, and messages $\textit{msg}_{i+1}^\pm, \ldots,\textit{msg}_k^\pm$
are the \emph{future messages} ($\textit{msg}_{i+1}^\pm$ is the immediate future
message).
Constraints can be also inserted into strands. 
A strand 
$[\textit{msg}_1^\pm, \linebreak[2]\ldots, \linebreak[2] \textit{msg}_k^\pm]$ 
is shorthand for 
$[nil ~|~ \textit{msg}_1^\pm,\linebreak[2] \ldots,\linebreak[2] \textit{msg}_k^\pm , nil ]$.
An \emph{initial state} is a state where the bar is at the beginning for all strands in the state,
and the network has no possible intruder fact of the form $\inI{\textit{m}}$. 
A \emph{final state} is a state where the bar is at the end for all strands in the state
and there is no negative intruder fact of the form $\nI{\textit{m}}$.

In the following example, we illustrate how the timed process algebra can be transformed into strands specifications of Maude-NPA.

\begin{example}\label{ex:brands-and-chaum-strand-time-trans}
The timed processes of Example~\ref{ex:brands-and-chaum-process-time}
are transformed into the following strand specification. 

\noindent
{\small
 \begin{align*}
(\mathit{Verifier}):\ [
&{-}(\textit{Commit}\ @\ A_1 : t'_1 \to V : t_1\uplus AS_1), \\[-.5ex]
& (t_1 = t'_1 + d(A_1,V) \wedge d(A_1,V)\geq 0), \\[-.5ex]
&{+}(n(V,f_1)\ @\ V : t_2 \to AS_2), \\[-.5ex]
&{-} ((n(V,f_1) \oplus N_P)\ @\ A_3 : t'_3 \to V : t_3\uplus AS_3), \\[-.5ex]
& (t_3 = t'_3 + d(A_3,V) \wedge d(A_3,V)\geq 0), \\[-.5ex]
& (t_3 \dot{-} t_2 \leq 2*d), \\[-.5ex]
&{-} (S_P\ @\ A_4: t'_4 \to V : t_4\uplus AS_4), \\[-.5ex]
& (t_4 = t'_4 + d(A_4,V) \wedge d(A_4,V)\geq 0), \\[-.5ex]
& \textit{open}(N_P,S_P,\textit{Commit}), \\[-.5ex]
& {-} (\textit{sign}(P,n(V,f_1) ; N_P \oplus n(V,f_1))@\ A_5: t'_5 \to V : t_5\uplus AS_5), \\[-.5ex] 
& (t_5 = t'_5 + d(A_5,V) \wedge d(A_5,V)\geq 0)
]
\\
(\mathit{Prover}):\  [
&{+}(\textit{commit}(n(P,f_1),s(P,f_2))@P : t_1 \to AS_1),\\[-.5ex]
&{-}(N_V@ \ A_2: t'_2 \to V : t_2\uplus AS_2), \\[-.5ex]
& (t_2 = t'_2 + d(A_2,P) \wedge d(A_2,P)\geq 0), \\[-.5ex]
&{+}((N_V \oplus n(P,f_1))@P : t_3 \to AS_3), \\[-.5ex]
&{+}(s(P,f_2)@P : t_4 \to AS_4), \\[-.5ex]
&{+}(\textit{sign}(P,N_V ; n(P,f_2) \oplus N_V)@P : t_5 \to AS_5)
]
 \end{align*}
}
 \end{example}

We specify the desired security properties in terms of \emph{attack patterns} including logical variables, which describe the insecure states that Maude-NPA is trying to prove unreachable.
Specifically, the tool attempts to find a \emph{backwards narrowing sequence} path from the attack pattern to an initial state 
until it can no longer form any backwards narrowing steps, at which point it terminates.
If   it has not found an initial state, the attack pattern is judged \emph{unreachable}. 

The following example shows how a classic mafia fraud attack for the Brands-Chaum protocol can be encoded in Maude-NPA's strand notation.

\begin{example}\label{ex:brands-and-chaum-strands-time-execution-mafia}
Following the strand specification of the Brands-Chaum protocol
given in Example~\ref{ex:brands-and-chaum-strand-time-trans},
the mafia attack of Example~\ref{ex:brands-and-chaum} is given as the following attack pattern.
Note that Maude-NPA uses symbol \texttt{===} for equality on the reals, \texttt{+=+} for addition on the reals, 
\texttt{*=*} for multiplication on the reals, 
and
\texttt{-=-} for subtraction on the reals. Also, we consider one prover \texttt{p}, one verifier \texttt{v}, and one intruder \texttt{i} at fixed locations.
Extra time constraints are included in an \texttt{smt} section,
where a triangular inequality has been added.
The mafia fraud attack is secure for Brands-Chaum and no initial state is found in the backwards search.

{
\scriptsize
\begin{verbatim}
       eq ATTACK-STATE(1) ---  Mafia fraud
          = :: r :: --- Verifier
            [ nil, -(commit(n(p,r1),s(p,r2))  @ i : t1 -> v : t2),
                    ((t2 === t1 +=+ d(i,v)) and d(i,v) >= 0/1),
                    +(n(v,r)                  @ v : t2 -> i : t2''),
                    -(n(v,r) * n(p,r1)        @ i : t3 -> v : t4),
                    (t3 >= t2 and (t4 === t3 +=+ d(i,v)) and d(i,v) >= 0/1),
                    ((t4 -=- t2) <= (2/1 *=* d)) | nil ] &
            :: r1,r2 :: --- Prover
            [ nil,  +(commit(n(p,r1),s(p,r2)) @ p : t1' -> i : t1''),
                    -(n(v,r) @ i : t2'' -> p : t3'),
                    ((t3' === t2'' +=+ d(i,p)) and d(i,p) >= 0/1),
                    +(n(v,r) * n(p,r1)        @ p : t3' -> i : t3'') | nil ]
            || smt(d(v,p) > 0/1 and d(i,p) > 0/1 and d(i,v) > 0/1 and d(v,i) <= d and 
                    (d(v,i) +=+ d(p,i)) >= d(v,p) and d(v,p) > d)
            || nil || nil || nil [nonexec] .
\end{verbatim}
}
\end{example}

\section{Experiments}\label{sec:experiments}
As a feasibility study, we have encoded several distance bounding protocols in Maude-NPA.
It was necessary to slightly alter the Maude-NPA tool by
(i) including minor modifications to the state space reduction techniques to allow for timed messages; 
(ii) the introduction of the sort \textsf{Real} and its associated operations;
and
(iii) the connection of Maude-NPA to a \emph{Satisfiability Modulo Theories (SMT)} solver\footnote{Several SMT solvers are publicly available, but the programming language Maude~\cite{maude-manual} currently  supports CVC4 \cite{cvc4} 
and
Yices \cite{yices}.} 
(see \cite{NOT06} for details on SMT).
The specifications, outputs, and the modified version of Maude-NPA
are available
at   {\small\url{http://personales.upv.es/sanesro/indocrypt2020/}}. 

Although the timed model allows an unbounded number of principals, the attack patterns used to specify insecure goal states allow us to limit the number of principals in a natural way.
In this case we  specified one  verifier, one prover, and one attacker, but allowed an unbounded number of sessions.

\begin{table}[h!]
\vspace{-.25cm}
\centering
\begin{tabular}{@{}l@{\ }c@{\ } |c@{\ }@{\ }c@{\ }|@{\ }c@{\ }@{\ }c@{\ }}
\toprule
Protocol
&
PreProc (sec)
&
Mafia
&
tm (sec)
&
Hijacking  & tm (sec) \\
\bottomrule
Brands and Chaum \cite{BC93}   &  3.0 &  $\checkmark$ & 4.3  & $\times$  & 11.4 \\
Meadows et al ($n_{V} \oplus n_{P}$,$P$) \cite{MPPC+07}&  3.7  &  $\checkmark$   & 1.3 & $\checkmark$  & 22.5     \\
Meadows et al ($n_{V}$,$n_{P} \oplus P$) \cite {MPPC+07}&  3.5 &  $\checkmark$   & 1.1 & $\times$   & 1.5  \\
Hancke and Kuhn \cite{HK05}&  1.2      & $\checkmark$    & 12.5   & $\checkmark$  & 0.7 \\
MAD   \cite{CBH03}&  5.1    &  $\checkmark$    &  110.5   & $\times$  & 318.8 \\
Swiss-Knife \cite{KAKS+08}&  3.1  &  $\checkmark$    & 4.8   & $\checkmark$  & 24.5  \\
Munilla et al. \cite{MP08}&  1.7 &  $\checkmark$    & 107.1 & $\checkmark$  & 4.5  \\
CRCS \cite{RC10}&  3.0 &  $\checkmark$   &  450.1   & $\times$ & 68.6  \\
TREAD \cite{ABGG+17}&  2.4 &  $\checkmark$      &  4.7  & $\times$ & 4.2  \\
\bottomrule
\end{tabular}
\vspace{.25cm}
\caption{Experiments performed for different distance-bounding protocols}
\label{fig:experiments}
\vspace{-5ex}
\end{table}

In Table~\ref{fig:experiments} above we present the results for the different distance-bounding protocols that we have analyzed.
Two attacks have been analyzed for each protocol:
a \emph{mafia fraud} attack
(i.e., an attacker tries to convince the verifier that an honest prover is closer to him than he really is),
and
a \emph{distance hijacking} attack
(i.e., a dishonest prover located far away succeeds in convincing a verifier that they are actually close, 
and he may only exploit the presence of honest participants in the neighborhood to achieve his goal).
Symbol $\checkmark$ means the property is satisfied and $\times$ means an attack was found.
The columns labelled $tm (sec)$ give the times in seconds that it took for a search to complete.  Finally the column labeled PreProc gives the time it takes Maude-NPA to perform some preprocessing on the specification that eliminates searches for some provably unreachable state.  This only needs to be done once, after which the results can be used for any query, so it is displayed separately. 

We note that, since our semantics is defined over  arbitrary metric spaces, not just Euclidean space, it is also necessary to verify that an attack returned by the tool is realizable over Euclidean space.  We note that the Mafia and hijacking  attacks  returned by Maude-NPA  in these experiments are all realizable on a straight line, and hence are  realizable over $n$-dimensional Euclidean space for any $n$. In general, this realizability check can be done via a final step in which the constraints with the Euclidean metric substituted for distance is checked via an SMT solver that supports checking quadratic constraints over the reals, such as Yices~\cite{yices}, Z3~\cite{z3}, or Mathematica~\cite{mathematica}.  Although this feature is not yet implemented in Maude-NPA, we have begun experimenting with these solvers. 

\vspace{-.25cm}
\section{Conclusions}\label{sec:conclusions}

We have developed a timed model for protocol analysis based on timing constraints,  and 
provided a prototype extension of Maude-NPA handling protocols
with time by taking advantage of
Maude's support of SMT solvers, as was done by Nigam et al. in \cite{NTU18}, and Maude-NPA's  support of constraint handling. We also performed some initial analyses to test the feasibility of the approach.  This approach should be applicable to other tools that support constraint handling.

There are several ways this work  can be extended.  One is to extend the ability of the tool to reason about a larger numbers or principals, in particular an unbounded  number of principals.   This includes an unbounded number of attackers; since each attacker must have its own location,  we cannot assume a single attacker as in Dolev-Yao.  Our specification and query language, and its semantics, supports reasoning about an unbounded number of principals, so this is a question of developing means of telling when a principal or state is redundant and developing state space reduction techniques based  on this.

Another important extension is to protocols that require the full Euclidean space model, in particular  those in  which location needs to be explicitly included in the constraints.   This includes for example protocols used for localization. For this, we have begun experimenting with SMT solvers that support solving quadratic constraints over the reals. 

Looking further afield, we consider adding different types of timing models.  In the timing model used in this paper, time is synonymous with distance.  But we may also be interested including other ways in which time is advanced, e.g. the amount of time a principal takes to perform internal processing tasks.  In our model, the method in which timing is advanced is specified by the $mte$ function, which is in turn used to generate constraints on which messages can be ordered.  Thus changing the way in which timing is advanced can be accomplished by modifying the $mte$ function.  Thus, potential future research includes design of   \emph{generic} $mte$   functions together with rules on their instantiation that guarantee soundness and completeness

Finally, there is also no reason for us to limit ourselves to time and location.  This approach  should be applicable to other  quantitative  properties as well.     For example, the inclusion of cost and utility would allow us to tackle new classes of problems not usually addressed by cryptographic protocol analysis tools, such as performance analyses (e.g., resistance against denial of service attacks), or even analysis of game-theoretic properties of protocols, thus opening up a whole new set of problems to explore.

\bibliographystyle{plain}

\newpage
\appendix

\section*{Additional Supporting Material}

In order to prove soundness and completeness of the transformation
in Appendix~\ref{proofs}, we first recall the untimed process algebra of Maude-NPA.

\section{
(Untimed) Process Algebra}\label{sec:PA}

Maude-NPA was originally defined \cite{EMM06,EMM09}
using \emph{strands} \cite{THG99}.
A process algebra that extends the strand
space model to naturally specify
protocols exhibiting choice points
was given in \cite{YEMM+16,YEMM19}.
Here we give a high-level summary of the untimed process algebra syntax of Maude-NPA, see
\cite{Maude-NPA}.

\subsection{Syntax of the Protocol Process Algebra }\label{sec:syntaxPA}

In the \emph{protocol process algebra} the behaviors of both honest principals
and the intruder are represented  by \emph{labeled processes}. 
Therefore, a protocol is specified as a set of labeled processes.
Each process performs a sequence of actions, namely, sending or receiving
a message, and may perform deterministic or non-deterministic choices. 
The protocol process algebra's syntax $\BNFSymbols$ is 
parameterized\footnote{More precisely, as explained in
Section \ref{sec:PASpecification},
they are parameterized
by a user-definable equational theory
$(\Sigma_{\mathcal{P}},E_{\mathcal{P}})$
having a sort \sort{Msg} of messages.}
 by a sort \sort{Msg} of messages and 
 a sort \sort{Cond} for conditional expressions.
 It has the following syntax:

{\small
\begin{align}
 \ProcConf~ &::= \LProc ~|~ \ProcConf ~\&~ \ProcConf ~|~ \emptyset \notag\\[-.8ex]
 \LProc ~&::= (\Role, I, J)~ \Proc \notag\\[-.8ex]
 \Proc~ &::=  \nil ~ |~ +\Msg~ | ~ -\Msg ~|~ \Proc \cdot \Proc  ~|  \notag \\[-.8ex]
              &~~~~~~~ \Proc~?~\Proc ~ | ~\textit{if} ~\Cond~ \textit{then} ~\Proc~ \textit{else} ~\Proc    \notag
\end{align}
} 

\begin{itemize}
	
\item $\ProcConf$ stands for a \emph{process configuration}, that is, a set of labeled processes. The symbol \& 
is used to denote set union for sets of labeled processes. 

\item $\LProc$ stands for a \emph{labeled process}, that is, a process $\Proc$ with a label $(\Role,I,J)$.
$\Role$ 
refers to the role of the process in the protocol (e.g., initiator or responder). $I$ is a 
natural number denoting the identity of the process, which distinguishes different instances(sessions) of a process specification.
$J$ indicates that the action at stage $J$ of the process specification 
will be the next one to be executed, that is, the first $J-1$ actions of the process for role $\Role$ have  already been executed. 
Note that we omit $I$ and $J$ in the protocol specification when both $I$ and $J$ are $0$. 

\item $\Proc$ defines the actions that can be executed within a process. ${+\Msg}$, and ${-\Msg}$ respectively denote
 sending out or receiving a message $\Msg$.  We assume a single channel, through which all messages are sent or received by the intruder.  
``$\Proc~\cdot~\Proc$" denotes \emph{sequential composition} of processes,
 where 
symbol \verb!_._! is associative and has the empty process $\nil$ 
as identity.
``$\Proc~?~\Proc$" denotes an explicit \emph{nondeterministic choice}, whereas 
``$\textit{if} ~\Cond~ \textit{then} \allowbreak ~\Proc~ \textit{else}
~\Proc$" denotes an explicit \emph{deterministic choice}, whose
continuation depends on the 
satisfaction of the constraint $\Cond$. 
In \cite{YEMM+16,YEMM19}, 
either equalities ($=$) or disequalities ($\neq$)  between message expressions were considered as constraints.

\end{itemize}

\vspace{-1ex}
Let 
$PS,~QS$, and $RS$ be process configurations,
and
$P,~Q$, and $R$ be protocol processes. 
The protocol syntax satisfies the following \emph{structural axioms}:

\begin{small}
\begin{minipage}{0.65\linewidth}
\begin{gather}
PS  \, \& \, QS = QS \, \& \, PS  \notag\\
(PS \, \& \, QS) \, \& \, RS = PS \, \& \, ( QS  \, \& \,  RS)\notag\\
(P \, \cdot \, Q) \cdot \, R  =  P \, \cdot \, (Q \, \cdot \, R)\notag
\end{gather}
\end{minipage}
\begin{minipage}{0.30\linewidth}
\begin{gather}
 PS \ \&\  \emptyset = PS  \notag\\
P \, \cdot \, \nil = P \notag\\
\nil \, \cdot \, P = P \notag
\end{gather}
\end{minipage}
\end{small}

The specification of the processes defining a protocol's behavior may contain 
some variables denoting information that 
the principal executing the process does not yet know, or that will be different in different executions. 
  In all protocol specifications we assume
three disjoint kinds of variables:
    \begin{itemize}
        	\item \textbf{\emph{fresh variables}}: these are not really variables in the standard sense,
        	 but \emph{names} for \emph{constant values} in a data type \sort{V_{fresh}}
        	 of unguessable values such as nonces. 
        	 A \emph{fresh variable} $f$ is always associated with a role $ro\in Role$ in the protocol.
        	Throughout this paper we will denote this kind of variables
        	as $f,f_1,f_2,\ldots$.

        	\item \textbf{\emph{choice variables}}: variables  first
        	appearing in a \emph{sent message} $\mathit{+M}$, which can be substituted 
        	by any value arbitrarily chosen from a possibly infinite domain.
        A choice variable indicates an  \emph{implicit non-deterministic choice}.
        Given a protocol with choice variables, each possible substitution of 
        these variables denotes a possible continuation of the protocol.
        	We always denote choice variables by uppercase
        	letters postfixed with the symbol ``?'' as a subscript, e.g., $\chV{A},\chV{B},\ldots$.

        	\item \textbf{\emph{pattern variables}}:  variables first appearing
        	in a \emph{received message} $\mathit{-M}$. These variables will be instantiated 
        	when matching sent and received messages. 
        	\emph{Implicit deterministic choices} are indicated by
                terms containing pattern variables, 
        	since failing to match a pattern term may lead to the rejection of a message.
        	A pattern term  plays the implicit role of a guard, 
        	so that, depending on the different ways of matching it, the protocol can have different continuations.
		    This kind of variables will be written with uppercase letters, e.g.,
		    $A,B,N_A,\ldots$.
        \end{itemize} 
        
Note that fresh variables are distinguished from other variables by having a specific sort \sort{Fresh}.
Choice variables or pattern variables can never have sort \sort{Fresh}.

We consider only \emph{well-formed} processes using the function
$\mathit{wf}: \mathit{\Proc} \rightarrow \mathit{Bool}$
defined in Figure~\ref{fig:wf}.
The definition of $\mathit{wf}$ uses an auxiliary function 
$\mathit{shVar} : \mathit{\Proc}  \rightarrow \mathit{VarSet}$
defined in Figure~\ref{fig:bvar}.
	Note that the well-formedness property implies that if a process begins with a deterministic choice
	action \textit{if T then Q else R}, then all variables in $T$ must be instantiated, and thus only one branch may be taken.  
	For this reason, it is undesirable to specify processes that begin with such an action.
	Furthermore, note that the well-formedness property  avoids 
	explicit choices where both possibilities
	are the $\nil$ process. That is, processes containing either \textit{(if T then nil else nil)}, or \textit{(nil ? nil)}, respectively.

\begin{figure}
{\small
\begin{align*}
& \wf(P \cdot +M) = \wf(P) 
\hspace{2ex} \textit{if} ~ (\var{M} \cap   \var{P}) \subseteq \bvar{P}\\[-.5ex]
  & \wf( P \cdot -M) = \wf(P)    
\hspace{2ex} \textit{if} ~ (\var{M} \cap   \var{P}) \subseteq \bvar{P}\\[-.5ex]
& \wf( P \cdot 
 (\textit{if} ~~T~ \textit{then} ~Q~ \textit{else} ~R))
  = \wf(P \cdot Q) \wedge \wf(P \cdot R) 
    \\
&\hspace{2ex}  \textit{if} \  P \neq \nil \ \textit{and } \ Q \neq \nil \ \textit{and }  \var{T} \subseteq \bvar{P}\\
&\wf(P \cdot (Q~?~R) )= \wf(P \cdot Q)  \wedge \wf(P \cdot R) 
\hspace{2ex} \textit{ if } Q \neq \nil \ \textit{or} R \neq \nil \\
&\wf(P \cdot ~\nil) = \wf(P) \\[-.5ex]
&\wf(\nil) = \textit{True}.
\end{align*}
}
\caption{The well-formed function}
\label{fig:wf}
\end{figure}

\begin{figure}
{\small
\begin{align*}
& \bvar{+M ~\cdot P} =  \var{M} \cup \bvar{P} \\[-.5ex]
& \bvar{-M~\cdot P} =   \var{M} \cup \bvar{P}  \\[-.5ex]
& \bvar{ (\textit{if} ~T~ \textit{then} ~P~ \textit{else} ~Q) ~\cdot R} \\[-.5ex]
& ~~~~~~ = 
    \var{T} \cup (\bvar{P} \cap \bvar{Q}) \cup \bvar{R}  \\[-.5ex]
& \bvar{ (P~?~Q)~\cdot R } =  ( \bvar{P} \cap \bvar{Q}) \cup \bvar{R} \\[-.5ex]
& \bvar{\nil} = \emptyset
\end{align*}
}
\caption{The shared variables function}
\label{fig:bvar}
\end{figure}

\subsection{Protocol Specification in Process Algebra} \label{sec:PASpecification}
We define a protocol $\caP$ in the 
protocol process algebra, written $\SpecPA$,
as a pair of the form $\SpecPA = ((\PASymbols,\PAPEq), \ProcPA)$,
where $(\PASymbols,\PAPEq)$ 
is the equational theory specifying the equational properties
of the cryptographic functions and the state structure,
and $\ProcPA$ is a term denoting a \emph{well-formed} process configuration representing
the behavior of the honest principals as well as the 
 capabilities of the attacker. That is, 
$\ProcPA = (\mathit{ro_1}) P_1 ~\&~ \ldots  ~\&~ (\mathit{ro_i}) P_i $,
where each $ro_k$, $1 \leq k \leq i$, 
is either the role of an honest principal or
identifies one of the capabilities of the attacker.
 $\ProcPA$ cannot contain two processes with the same label, 
i.e., the behavior of each honest principal, and each attacker capability
are represented by a \emph{unique} process. $\PAPEq =\PEq \cup \PAEq$ is a set of equations with $\PEq$ denoting the protocol's cryptographic properties and $\PAEq$ denoting the properties of process constructors. 
The set of equations $\PEq$ is user-definable and can
vary for different protocols.   Instead,
the set of equations $\PAEq$ is always the same for all protocols. 
$\PASymbols=\PSymbols \cup \BNFSymbols$ is the signature
defining the sorts and function symbols as follows:
 
 \begin{itemize}
 \item $\PSymbols$ is an order-sorted signature defining the sorts and function symbols
        for the messages that can be exchanged in protocol $\caP$.
       However, independently of protocol $\caP$, $\PSymbols$  must always
        have  two sorts  \sort{Msg} and \sort{Cond} as the top sorts in one of its 
       connected components.

 \item $\BNFSymbols$ is an order-sorted signature defining the sorts 
        and function symbols of the \emph{process algebra infrastructure}. 
        $\BNFSymbols$ corresponds
       exactly to the BNF definition 
       of the protocol process algebra's syntax in Section \ref{sec:syntaxPA}. 
\end{itemize}
Therefore, the syntax $\PASymbols$ of processes for $\caP$ will be in the union signature $\BNFSymbols \cup \PSymbols$,
consisting of the protocol-specific syntax $\PSymbols$, and the generic process syntax $\BNFSymbols$ through the shared sort \sort{Msg}.

\subsection{Process Algebra Semantics}\label{sec:semanticsPA}

Given a protocol $\mathcal{P}$, 
a \emph{state} of $\caP$ consists of a set 
of (possibly partially executed) \emph{labeled processes}, and a set of terms in the intruder knowledge $\{IK\}$.
That is, a state is a term of the form n
$\{ LP_1 \,\&\, \cdots \,\&\, LP_n ~|~ \{\textit{IK}\}\}$.
Given a state $St$ of this form,  we abuse notation and write 
$LP_k \in St$ if $LP_k$ is a labeled process in the set $LP_1 \,\&\, \cdots \,\&\, LP_n$.

The intruder knowledge \textit{IK} models the \emph{single} channel
through which all messages are sent and received.
Messages are stored in the form $\inI{M}$.
We consider an active attacker who has complete control of the channel, 
i.e, can read, alter, redirect, and delete traffic as well as create
its own messages by means of
\emph{intruder processes}.  That is, the purpose of some $LP_k \in St$
is to perform message-manipulation actions for the intruder.

State changes are defined by a set $\PARls$ of \emph{rewrite rules}, such that
the rewrite theory $(\PAStateSymbols, \PAPEq , \PARls)$ characterizes the behavior of 
protocol $\caP$, where  $\PAStateSymbols$ extends $\PASymbols$
by adding state constructor symbols.
We assume that a protocol's
execution begins with an empty state, i.e., a state with an empty
set of labeled processes, and an empty network. 
That is, the initial state is always of the form 
$ \{ \emptyset ~|~ \{\emptyset\}  \}$.
Each transition rule in $\PARls$ is labeled with a tuple of the form $\mathit{(ro,i, j, a,n)}$, where:

\begin{itemize}
	\item $\mathit{ro}$ is  the role of the labeled process being executed in the transition.
     
    \item $i$ denotes the identifier of the labeled process being executed in the transition. Since there can be more than one process instance of the same role in a process state, $i$ is used to distinguish different instances, i.e., $ro$ and $i$ together uniquely identify a process in a state.
	
	\item $j$ denotes the process' step number since its beginning.

	\item $a$ is a ground term identifying the action that is being performed in the transition.
		  It has different possible values: 
		    ``$+m$'' or ``$-m$'' if the message $m$ was sent  (and added to the network) or received, respectively;
		    ``$m$'' if the message $m$ was sent  but did not increase the network,
		    ``$?$'' if the transition performs an explicit non-deterministic choice, or ``$\mathit{T}$'' if the transition 
		  performs an explicit deterministic choice.
	\item $n$ is a number that, if the action that is being executed is an explicit choice, indicates which branch has been chosen as the process continuation.
		  In this case $n$ takes the value of either $1$ or $2$.
	      If the transition does not perform any explicit choice, then $n=0$.
\end{itemize}

The set $\PARls$ of transition rules 
that define the 
execution of a state are given in \cite{YEMM+16,YEMM19}.
Note that in the transition rules $\PARls$ shown below, $PS$ denotes the rest of labeled processes of the state (which can be 
the empty set $\emptyset$).

\begin{itemize}
\item 
The action of \emph{sending a message} is represented by the two transition rules
below. Since we assume the intruder has complete control of the network, 
it can learn any message sent by other principals.
Rule~\eqref{eq:pa-output-modIK} denotes the case in which the
sent message is added  to the intruder knowledge.
Note that this rule can only be applied if the intruder has not
already learnt that message.
Rule~\eqref{eq:pa-output-noModIK}
denotes the case in which the intruder chooses not to learn the message, i.e., the intruder knowledge is not modified,
and, thus, no condition needs to be checked. 
Since choice variables denote messages that are nondeterministically chosen, all (possibly infinitely many) admissible ground substitutions for the choice variables are possible behaviors.

\begin{small}
 \begin{align}
&\{ (ro,i, j)~( +M \cdot P) ~\&~ PS  ~|~ \{IK\}  \} \notag\\[-.5ex]
&  \longrightarrow_{(ro,i,j,+M\sigma,0)}
\{ (ro,i, j+1)~P \sigma ~\&~ PS ~|~ \{ \inI{M\sigma}, IK\}  \} \notag\\[-.5ex]
&  \textit{ if } ( \inI{M\sigma}) \notin \textit{IK}  \notag\\[-.5ex]
&  \textit{where } \sigma  \ 
 \textit{is a ground substitution}
 \textit{ binding choice variables}  
\textit{ in} \  M   \eqname{{\small PA++}}
 \label{eq:pa-output-modIK}
  \end{align}

  \begin{align}
  &\{ (ro,i, j)~( +M \cdot P) ~\&~ PS  ~|~ \{IK\}  \} \notag\\[-.5ex]
  &  \longrightarrow_{(ro, i, j, M\sigma,0)}
 \{ (ro,i, j+1)~P\sigma ~\&~ PS ~|~ \{IK\} \} \notag\\[-.5ex]
 &  \textit{where } \sigma 
  \ \textit{is a ground substitution} 
\  \textit{binding choice variables} 
\textit{ in} \  M  \eqname{PA+}
 \label{eq:pa-output-noModIK}
 \end{align}
\end{small}

\item As shown in the rule below, a process can \emph{receive a message} matching a pattern $M$ if
there  is a message $M'$ in the intruder knowledge, i.e., a message
previously sent either by some honest principal or by some intruder process, 
that matches the  pattern message $M$. 
After receiving this message the process
will continue with its variables instantiated by the matching substitution, 
which takes place modulo the equations $E_\caP$. Note that the intruder can ``delete" a message   via choosing not to learn it (executing Rule \ref{eq:pa-output-noModIK} instead of Rule \ref{eq:pa-output-modIK})  or not
to deliver it (failing to execute Rule \ref{eq:pa-input}).

\begin{small}
\begin{align}
&\{	(ro,i, j)~( -M\cdot P) ~\&~ PS \mid \{\inI{M'}, IK\} \}\notag\\[-.5ex]
&	\longrightarrow_{(ro,i,j,-M\sigma,0)}
\{	(ro,i, j+1)~P\sigma ~\&~ PS \mid \{ \inI{M'}, IK\}  \} \notag\\[-.5ex]
&~\textit{if}~ M'=_{\PEq} M\sigma  \eqname{PA-}
	\label{eq:pa-input}
\end{align}
\end{small}

\item The two transition rules shown below define the operational semantics of 
\emph{explicit deterministic choice}s. That is, the 
 operational semantics of an
  $\textit{if} ~T~ \textit{then} ~P~ \allowbreak \textit{else} ~Q$ expression.
More specifically, rule~\eqref{eq:pa-detBranch1} describes  the
\textit{then} case, i.e., if the constraint $T$ is satisfied, 
the process will continue as $P$.
Rule~\eqref{eq:pa-detBranch2} describes the \textit{else} case, 
that is, if the constraint $T$ is \emph{not} satisfied, the process will continue as $Q$.
Note that, since we only consider well-formed processes, these transition
rules will only be applied if $j \ge 1$. Note also that since $T$ has been fully substituted by the time the if-then-else is executed, 
the validity 
of $T$ can be easily checked. 

\begin{small}
\begin{align}
 & \{ (ro,i, j)~((\textit{if} ~T~ \textit{then} ~P~ \textit{else} ~Q) ~\cdot R) ~\&~ PS \mid \{IK\}\}\notag\\[-.5ex]
 & \longrightarrow_{(ro,i, j, T,1)}
\{ (ro,i, j+1)~(P\cdot R) ~\&~ PS \mid \{IK\} \}
 ~~if~T  
 \eqname{PAif1}
 \label{eq:pa-detBranch1}
\end{align}

\begin{align}
& \{ (ro,i, j)~( (\textit{if} ~T~ \textit{then} ~P~ \textit{else} ~Q) ~\cdot R) ~\&~ PS \mid \{IK\} \} \notag\\[-.5ex]
& \longrightarrow_{(ro,i,j,T,2)}  
\{ (ro,i, j+1)~(Q\cdot R) ~\&~ PS \mid \{IK\} \}
 ~~if~  \neg T 
 \eqname{PAif2}
  \label{eq:pa-detBranch2}
\end{align}
\end{small}

\item The two transition rules below define the semantics of 
\emph{explicit non-deterministic choice} $P~?~Q$. In this case, 
the process can continue either as $P$, denoted by   rule~\eqref{eq:pa-nonDetBranch1},
or as $Q$, denoted by  rule~\eqref{eq:pa-nonDetBranch2}. 
Note that this decision is made non-deterministically. 

\begin{small}
\begin{align}
& \{ (ro,i, j)~((P~?~Q)\cdot R) ~\&~ PS \mid \{IK\} \}\notag\\[-.5ex]
&\longrightarrow_{(ro,i,j,?,1)} 
\{ (ro,i, j+1)~(P \cdot R) ~\&~ PS \mid \{IK\}  \} 
\eqname{PA?1}
\label{eq:pa-nonDetBranch1}\\[.5ex]
& \{ (ro,i, j)~((P~?~Q)\cdot R) ~\&~ PS \mid \{IK\} \}\notag\\[-.5ex]
&\longrightarrow_{(ro,i,j,?,2)} 
\{ (ro,i,j+1)(Q \cdot R) ~\&~ PS \mid \{IK\} \} 
\eqname{PA?2}
\label{eq:pa-nonDetBranch2}
\end{align}
\end{small}

\item 
The transition rules shown below describe the \emph{introduction of a new 
process} from the specification into the state, 
which allows us to support an unbounded session 
model. 
Recall that fresh variables are associated with a role and an identifier. 
Therefore, whenever a new process is introduced: 
(a) the largest process identifier $(i)$ will be increased by 1, and 
(b) new names will be assigned to the fresh variables in the new process.

\vspace{-1.5ex}
\begin{small}
 \begin{align}
 \left \{
 \begin{array}{@{}l@{}}
  \forall \  (ro)~ P_k \in \ProcPA\notag\\[.5ex]
  \{ PS \mid \{IK\} \} \notag\ \\[-.5ex]
  \longrightarrow_{(ro, i+1, 1,A,Num)}
 \{ (ro, i+1, 2)~P'_k ~\&~ PS \mid \{IK'\}  \}\notag\\ [1ex]
  \textsf{IF} ~  \{(ro,i+1,1)~ P_k\rho_{ro,i+1} \mid \{IK\}  \} \notag\\[-.5ex]
\longrightarrow_{(ro,i+1, 1,A,Num)} 
 \{	(ro,i+1, 2)~P'_k ~ \mid \{IK'\} \} \\[1ex]
  \textit{where } \rho_\mathit{ro,i+1} 
 \ \textit{is a fresh substitution}, \notag\\
 \  i= \MaxProcId(PS, ro)  
 \end{array}
 \right \}\eqname{PA\&}
 \label{eq:pa-new}
 \end{align}
 \end{small}
 \vspace{-1.5ex}

 \noindent
 Note that $A$ denotes the action of the state transition,
 and can be of any of the forms explained above.
 The function $\MaxProcId$ is defined as follows:

\vspace{-1.5ex}
\begin{small}
\begin{align*}
& \MaxProcId(\emptyset, ro) = 0 \notag \\
& \MaxProcId((ro, i, j) P \& PS, ro)  = max(\MaxProcId(PS, ro), i) \notag \\
& \MaxProcId((ro', i, j) P \& PS, ro) = \MaxProcId(PS, ro) 
\ \ \ \textit{if} \ ro \neq ro' \notag 
\end{align*}
\end{small}
\vspace{-1.5ex}

\noindent
where $PS$ denotes a process configuration, $P$ a process, and $ro, ro'$  role names.
\end{itemize} 
 
 Therefore, the behavior of a protocol in the process algebra is
 defined by the set of transition rules
 $\PARls = \{ \eqref{eq:pa-output-modIK},\allowbreak  \eqref{eq:pa-output-noModIK},\allowbreak 
              \eqref{eq:pa-input},\allowbreak
              \eqref{eq:pa-detBranch1},\allowbreak    \eqref{eq:pa-detBranch2},\allowbreak
              \eqref{eq:pa-nonDetBranch1}, \allowbreak  \eqref{eq:pa-nonDetBranch2} \} \cup
              \eqref{eq:pa-new} $.
              
The main result in \cite{YEMM+16,YEMM19} is a 
bisimulation between the strand state space generated by 
the narrowing-based backwards semantics of \cite{EMM06,EMM09}
and 
the transition rules $\PARls$ above, associated to the forwards semantics for process algebra.
This is nontrivial, since there are three major
 ways in which the two semantics differ.  
 The first is that  processes  ``forget'' their past,  while strands ``remember'' theirs.  
 The second is that Maude-NPA uses backwards search, while the process
 algebra proceeds forwards.  
 The third is that Maude-NPA performs symbolic reachability analysis using
 terms with variables, while the process algebra considers only ground terms. 

\section{Soundness and Completeness Proofs}\label{proofs}

The problem with adding choice variables $\chV{t}$ and $\chV{AS}$ to untimed sending messages is that 
they may be replaced by values that do not have a counterpart at the timed process algebra.
Let us clarify the two relevant sets of states.

\begin{definition}[TPA-State]
Given a protocol $\caP$, 
its time  process specification
$((\TPASymbols,\allowbreak\TPAPEq),\allowbreak \ProcTPA)$
and
its associated rewrite theory $(\TPAStateSymbols,\allowbreak \TPAPEq,\allowbreak \TPARls)$,
a \emph{TPA-State} 
 is a state in the time  process algebra semantics
that is \emph{reachable} from the initial state 
$\{\emptyset ~|~ \{\emptyset\} ~|~ 0.0\}$. 
\end{definition}

\begin{definition}[PA-State]
Given a protocol $\caP$, 
its time  process specification
$((\TPASymbols,\allowbreak\TPAPEq),\allowbreak \ProcTPA)$,
the simplified version
$((\PASymbols,\allowbreak\PAPEq),\allowbreak \tpatopa(\ProcTPA))$,
and its associated rewrite theory $(\PAStateSymbols,\allowbreak \PAPEq,\allowbreak \PARls)$,
a \emph{PA-State} 
 is a state in the process algebra semantics
that is \emph{reachable} from the initial state 
$\{\emptyset ~|~ \{\emptyset\} \}$. 
\end{definition}

We consider \emph{successful transition sequences} where the additional conditional expressions are evaluated.
              
\begin{definition}[Successful PA-state]
Given a protocol $\caP$, 
we say a PA-state is \emph{successful} if for each process $(ro,i, j) P_1 \cdots P_n$ in the state, 
the first action $P_1$ is not a conditional expression 
introduced by the transformation $\tpatopa$,
and the conditions of
all the conditional expressions 
introduced by the transformation $\tpatopa$ were evaluated to true in the sequence reaching the PA-state from the initial PA-state.
\end{definition}              

The transition rule
\eqref{eq:time} does not have a counterpart without time. 
The transition rule \eqref{eq:time} is related to a proper interleaving of input actions, in such a way that closer participants receive a message earlier than others.

\begin{definition}[Realizable PA-state]\label{def:realizable}
Given a protocol $\caP$, 
a sequence of PA-states

\noindent
\begin{align*}
\{ \emptyset ~|~ \{\emptyset\}\} &\longrightarrow_{(\textit{ro}_1,i_1,j_1,A_1,k_1)} \textit{PA}_1 \\
&\longrightarrow_{(\textit{ro}_2,i_2,j_2,A_2,k_2)} \textit{PA}_2 \\[-2ex]
&\hspace{3ex} \vdots\\[-2ex]
&\longrightarrow_{(\textit{ro}_n,i_n,j_n,A_n,k_n)} \textit{PA}_n 
\end{align*}

\noindent
is called \emph{realizable} if the following two sets of distance constraints are satisfied:

\begin{itemize}
\item \emph{Triangular inequalities}. 
For all distinct $1 \leq j_1,j_2,j_3 \leq n$, $d((ro_{j_1},i_{j_1}),(ro_{j_2},i_{j_2})) \leq d((ro_{j_1},i_{j_1}),(ro_{j_3},i_{j_3})) + d((ro_{j_3},i_{j_3}),(ro_{j_2},i_{j_2}))$.
\item \emph{Time sequence monotonicity}.
For $1\leq j\leq n$ 
such that
$A_j = -(M @ \ldots)$
and $M\ @\ ({(ro'_0,i'_0) : t} \to ((ro'_1,i'_1) : t_1 \cdots (ro_j,i_j): t_j \cdots (ro'_m,i'_m) : t_m))$ 
is stored in the \textit{Net} component of $\textit{PA}_j$,
$d((ro'_0,i'_0),(ro'_k,i_k)) \leq d((ro'_0,i'_0),(ro_{j},i_{j}))$,
for all $1\leq k\leq j$.
\end{itemize}
We say a PA-state is \emph{realizable} if 
the sequence reaching the PA-state from the initial PA-state is.
\end{definition}

\begin{lemma}[Realizable PA-state]
Given a protocol $\caP$ and
a \emph{realizable} sequence of PA-states

\noindent
\begin{align*}
\{ \emptyset ~|~ \{\emptyset\}\} &\longrightarrow_{(\textit{ro}_1,i_1,j_1,A_1,k_1)} \textit{PA}_1 \\
&\longrightarrow_{(\textit{ro}_2,i_2,j_2,A_2,k_2)} \textit{PA}_2 \\[-2ex]
&\hspace{3ex} \vdots\\[-2ex]
&\longrightarrow_{(\textit{ro}_n,i_n,j_n,A_n,k_n)} \textit{PA}_n 
\end{align*}

\noindent
there exists a sequence of TPA-states

\noindent
\begin{align*}
\{ \emptyset ~|~ \{\emptyset\}  ~|~ 0.0\} &\longrightarrow_{(\textit{ro}'_1,i'_1,j'_1,A'_1,k'_1,t'_1)} \textit{TPA}_1 \\
&\longrightarrow_{(\textit{ro}'_2,i'_2,j'_2,A'_2,k'_2,t'_2)} \textit{TPA}_2 \\[-2ex]
&\hspace{3ex} \vdots\\[-2ex]
&\longrightarrow_{(\textit{ro}'_{n'},i'_{n'},j'_{n'},A'_{n'},k'_{n'},t'_{n'})} \textit{TPA}_{n'} 
\end{align*}

\noindent
such that for every two input steps 
$\textit{PA}_a \longrightarrow_{(\textit{ro},i,j_a,-(m @ \ldots),0)} \textit{PA}_{a+1}$,
$a\in\{0,\ldots,n-1\}$,
$\textit{PA}_b \longrightarrow_{(\textit{ro}',i',j_b,-(m'@\ldots),0)} \textit{PA}_{b+1}$,
$b\in\{0,\ldots,\allowbreak n-1\}$
s.t. $a < b$,
there exits 
two input steps with the same roles 
$(\textit{ro},i)$ and $(\textit{ro}',i')$
and messages $m$ and $m'$,
$\textit{TPA}_c \longrightarrow_{(\textit{ro},i,j_c,-(m),0,t_c)} \textit{TPA}_{c+1}$,
$c\in\{0,\ldots,{n'}-1\}$,
$\textit{TPA}_d \longrightarrow_{(\textit{ro}',i',j_d,-(m'),0,t_d)} \textit{TPA}_{d+1}$,
$d\in\{0,\ldots,{n'}-1\}$
s.t. $c < d$, which implies that $t_c < t_d$.
\end{lemma}

The transition rule \eqref{eq:tpa-new} corresponds to \eqref{eq:pa-new} but they are very different.
On the one hand, the transition rule \eqref{eq:pa-new} adds a new process if it starts with either an output message, a conditional, a non-deterministic choice, 
or an input message that synchronizes with the intruder knowledge.
On the other hand, the transition rule \eqref{eq:tpa-new} adds a new process without advancing it.
That is, it can add a time  process starting with an input message that cannot be synchronized with the current intruder knowledge.
In this case, the process never moves forward, and so it can be ignored.

\begin{definition}[Blocked Process]
Given a protocol $\caP$
and a TPA-state 
\linebreak $\{ TLP_1 \,\&\, \cdots \,\&\, TLP_n ~|~ \{\textit{IK}\}\}$, 
we say a time  process $TLP_i=(ro,i, j)\ P_i$ is \emph{blocked} if 
there is no transition step from 
$ \{(ro,i,j)~ P_i \mid \{IK\}  \}$.
\end{definition}

\begin{definition}[Non-void TPA-state]
Given a protocol $\caP$, 
we say a TPA-state is \emph{non-void} if 
the last transition step in the sequence reaching the state was not \eqref{eq:tpa-new}.
\end{definition}              

We now 
define the relation $\caH$ that relates PA and TPA states. 

\begin{definition}[Relation $\caH$]
 Given a protocol $\caP$,
a TPA-State $\textit{TPA}=\{TLP_1 \& \allowbreak \ldots \& TLP_n \mid \{ Net \} \mid \bar{t} \}$ 
 and 
a successful PA-State $\textit{PA}=\{LP_1 \& \allowbreak \ldots \& LP_m \mid \{ IK \} \}$ ,
we have that $(\textit{TPA},\textit{PA}) \in \caH$ iff:
\begin{itemize}
 	\item[(i)]  For each non-blocked timed process $TLP_k=(ro,i, j)~P_k $, $1\leq k \leq n$,
 	there exists a process $LP_{k'}=(ro,i, j')~P'_k$, $1\leq k' \leq m$,
 	such that $P'_k=\tpatopax(P_k,ro,i)$; and viceversa.
 	\item[(ii)] For each stored message $(M@((\textit{ro},i) : t \to AS))$ in \textit{Net},
	there exists $\inI{(M@((\textit{ro},i) : t \to AS \uplus AS'))}$ in $IK$; and viceversa.
 \end{itemize} 
\end{definition}

We are able to prove soundness and completeness.

\begin{proposition}[Completeness]\label{prp:comp}
Given a protocol $\caP$
and a non-void TPA-state \textit{TPA}, there exists a successful PA-state \textit{PA} s.t. $\textit{TPA}\ \caH\ \textit{PA}$.
\end{proposition}
\begin{proof}
By induction on the length $n$ of 

\noindent
\begin{align*}
\{ \emptyset ~|~ \{\emptyset\}  ~|~ 0\} &\longrightarrow_{(\textit{ro}_1,i_1,j_1,A_1,k_1,t_1)} \textit{TPA}_1 \\
&\longrightarrow_{(\textit{ro}_2,i_2,j_2,A_2,k_2,t_2)} \textit{TPA}_2 \\[-2ex]
&\hspace{3ex} \vdots\\[-2ex]
&\longrightarrow_{(\textit{ro}_{n},i_{n},j_{n},A_{n},k_{n},t_{n})} \textit{TPA}_{n} 
\end{align*}

\noindent
If $n=0$, then the conclusion follows.
If $n>0$, then the induction hypothesis says that there exists a sequence

\noindent
\begin{align*}
\{ \emptyset ~|~ \{\emptyset\} \} &\longrightarrow_{(\textit{ro}'_1,i'_1,j'_1,A'_1,k'_1)} \textit{PA}_1 \\
&\longrightarrow_{(\textit{ro}'_2,i'_2,j'_2,A'_2,k'_2)} \textit{PA}_2 \\[-2ex]
&\hspace{3ex} \vdots\\[-2ex]
&\longrightarrow_{(\textit{ro`}_{m-1},i'_{m-1},j'_{m-1},A'_{m-1},k'_{m-1})} \textit{PA}_{m-1} 
\end{align*}

\noindent
such that $\textit{TPA}_{n-1}\ \caH\ \textit{PA}_{m-1}$.
Let us consider the transition rule used in the step $n$.
\begin{itemize}
\item Transition rules               \eqref{eq:tpa-detBranch1},\allowbreak  \eqref{eq:tpa-detBranch2},\allowbreak 
              \eqref{eq:tpa-nonDetBranch1},\allowbreak  and \eqref{eq:tpa-nonDetBranch2}  are immediate, since they do not use any information from the network
              or the global time $\bar{t}$.
\item Transition rules               \eqref{eq:tpa-output-modIK} and \eqref{eq:tpa-output-noModIK} are also immediate, since they simply add extra information to the
network. 
\item The transition rule \eqref{eq:tpa-input} is also immediate because the network of 
$\textit{TPA}_{n-1}$ and the intruder knowledge of $\textit{PA}_{m-1}$ contain the same number of messages. 
Clearly, we assume that the choice variables 
$\chV{t}$ and $\chV{AS}$ for time added by $\tpatopa$ are replaced by the very same information existing in $\textit{TPA}_{n-1}$. In this case, since we consider only successful states, the conditional expressions added by $\tpatopa$ are evaluated to true.
\item The transition rule \eqref{eq:time} is also immediate, since the global time is the only change. 
\item The transition rule \eqref{eq:tpa-new} is excluded by requiring it to be non-void. \qed
\end{itemize}
\end{proof}

\begin{proposition}[Soundness]\label{prp:snd}
Given a protocol $\caP$
and a successful realizable PA-state \textit{PA}, there exists a non-void TPA-state \textit{TPA} s.t. $\textit{TPA}\ \caH\ \textit{PA}$.
\end{proposition}
\begin{proof}
By induction on the length $m$ of 

\noindent
\begin{align*}
\{ \emptyset ~|~ \{\emptyset\} \} &\longrightarrow_{(\textit{ro}_1,i_1,j_1,A_1,k_1)} \textit{PA}_1 \\
&\longrightarrow_{(\textit{ro}_2,i_2,j_2,A_2,k_2)} \textit{PA}_2 \\[-2ex]
&\hspace{3ex} \vdots\\[-2ex]
&\longrightarrow_{(\textit{ro}_{m},i_{m},j_{m},A_{m},k_{m})} \textit{PA}_{m} 
\end{align*}

\noindent
If $n=0$, then the conclusion follows.
If $n>0$, then the induction hypothesis says that there exists a sequence

\noindent
\begin{align*}
\{ \emptyset ~|~ \{\emptyset\}  ~|~ 0\} &\longrightarrow_{(\textit{ro}'_1,i'_1,j'_1,A'_1,k'_1,t'_1)} \textit{TPA}_1 \\
&\longrightarrow_{(\textit{ro}'_2,i'_2,j'_2,A'_2,k'_2,t'_2)} \textit{TPA}_2 \\[-2ex]
&\hspace{3ex} \vdots\\[-2ex]
&\longrightarrow_{(\textit{ro}'_{n-1},i'_{n-1},j'_{n-1},A'_{n-1},k'_{n-1},t'_{n-1})} \textit{TPA}_{n-1} 
\end{align*}

\noindent
such that $\textit{TPA}_{n-1}\ \caH\ \textit{PA}_{m-1}$.
Let us consider the transition rule used in the step $m$.
\begin{itemize}
\item Transition rules               \eqref{eq:pa-detBranch1},\allowbreak  \eqref{eq:pa-detBranch2},\allowbreak 
              \eqref{eq:pa-nonDetBranch1},\allowbreak  and \eqref{eq:pa-nonDetBranch2}  are immediate.
\item Transition rules               \eqref{eq:pa-output-modIK}, \eqref{eq:pa-output-noModIK}, and \eqref{eq:pa-input} are also immediate, since they simply add the message to the intruder knowledge and the choice variables $\chV{t}$ and $\chV{AS}$ added by $\tpatopa$ are replaced by valid values, since the PA-state is successful and realizable.
\item The transition rule \eqref{eq:pa-new} is also immediate because implies two transition steps from $\textit{TPA}_{n-1}$, one using the transition rule 
\eqref{eq:tpa-new} and the other one for the very same action of the step \eqref{eq:pa-new}. \qed
\end{itemize}
\end{proof}

\end{document}